\documentclass[a4paper,UKenglish,cleveref, autoref, thm-restate]{lipics-v2021}

\title{Greedy Vector Balancing}

\author{Wojciech Czerwi\'nski}{University of Warsaw, Poland}{wczerwin@mimuw.edu.pl}{https://orcid.org/0000-0002-6169-868X}
{Supported by the ERC grant INFSYS, agreement no. 950398}
\author{Daniel Dadush}{Centrum Wiskunde en Informatica \& Utrecht University, The Netherlands}{d.n.dadush@cwi.nl}{https://orcid.org/0000-0001-5577-5012}
{}
\author{Ekin Ergen}{Technical University of Berlin, Germany}{ergen@math.tu-berlin.de}{https://orcid.org/0009-0000-5502-2619}
{Supported by DFG under MATH+
(EXC-2046/2, project ID: 390685689)}
\author{Arka Ghosh}{University of Bordeaux, France}{arka.ghosh@u-bordeaux.fr}{https://orcid.org/0000-0003-3839-8459}
{Partially supported by the Polish National Science Centre (NCN) grant ”Linear algebra in orbit finite dimension” (2022/45/N/ST6/03242) and the SAIF project, funded by the ”France 2030” government
investment plan managed by the French National Research Agency, under the reference ANR-23-PEIA-0006}
\author{S{\l}awomir Lasota}{University of Warsaw, Poland}{s.lasota@uw.edu.pl}{https://orcid.org/0000-0001-8674-4470}
{Partially supported by the ERC grant INFSYS, agreement no. 950398, and by the NCN grant 2024/55/B/ST6/01674.}
\author{{\L}ukasz Orlikowski}{University of Warsaw, Poland}{l.orlikowski@mimuw.edu.pl}{https://orcid.org/0009-0001-4727-2068}
{Supported by the ERC grant INFSYS, agreement no. 950398}

\authorrunning{Czerwi\'nski et al.} 



\EventEditors{John Q. Open and Joan R. Access}
\EventNoEds{2}
\EventLongTitle{42nd Conference on Very Important Topics (CVIT 2016)}
\EventShortTitle{CVIT 2016}
\EventAcronym{CVIT}
\EventYear{2016}
\EventDate{December 24--27, 2016}
\EventLocation{Little Whinging, United Kingdom}
\EventLogo{}
\SeriesVolume{42}
\ArticleNo{23}

\usepackage{xcolor}
\usepackage{xspace}
\usepackage[disable]{todonotes}
\usepackage{mathtools}
\usepackage{amsthm}
\usepackage{enumitem}
\usepackage{cleveref}
\usetikzlibrary{arrows.meta}

\newcommand{\R}{\mathbb{R}}

\newcommand{\N}{\mathbb{N}}

\newcommand{\vr}[1]{\mathbf{#1}}
 
\newcommand{\spn}[1]{\textsc{span}(#1)}
\newcommand{\setof}[2]{\left\{#1 \, : \, #2\right\}}
\newcommand{\para}[1]{\vspace{-2mm}\subparagraph*{\rm \bf #1.}}
\newcommand{\dist}[1]{{\rm dist}({#1})}

\newcommand{\innerprod}[2]{\langle #1, #2 \rangle}

\newcommand{\interval}[2]{[#1,#2]}
\newcommand{\convhull}[1]{\textsc{conv-hull} #1}
\newcommand{\proj}[1]{\pi_{#1}}

\newcommand{\ball}[2]{{\cal B}^d(#2)}

\newcommand{\eps}{\bar\delta}

\newcommand{\hypp}[1]{{\cal H}^\bot_{#1}}
\newcommand{\KK}[1]{\KKK^{\delta}_{#1}}
\newcommand{\KKK}{{\cal K}}
\newcommand{\slball}[2]{{\cal B}^m(#2)}
\newcommand{\slsphere}[2]{{\cal S}^{m-1}(#2)}

\renewcommand{\dim}[1]{\text{dim}(#1)}

\DeclarePairedDelimiter{\norm}{\lVert}{\rVert}

\hideLIPIcs
\nolinenumbers
\begin{document}

\maketitle


\begin{abstract}

In online vector balancing, vectors $\vr t_1, \dots, \vr t_n$ arrive one
by one from a given set $T$ and the goal is to assign signs $s_1,\dots,s_n \in \{\pm 1\}$ in an
online manner so as to minimize the largest norm of any signed prefix sum $\sum_{i=1}^k s_i \vr
t_i$, $k \in [n]$. In this paper, we
analyze the natural Euclidean \emph{greedy vector balancing algorithm} for this problem: at each step
$k$, the sign $s_k \in \{\pm 1\}$ is chosen so that $s_k \vr t_k$ has non-positive inner product with $\sum_{i=1}^{k-1} s_i \vr t_i$. Our main result is the first finite bound, independent of the sequence length $n$, on the performance of greedy whenever $T$ is finite. When $T \subset \R^d$ consists of unit vectors, we prove that the signed sums produced by greedy have Euclidean norm at most $(2/\delta_T)^{d-1}$,
where $\delta_T$ is the minimum non-zero distance between vectors in $T$ and
subspaces spanned by vectors in $T$. The same upper bound holds when the
sequences are composed of scaled down vectors in $T$. We also provide a simple
set $T$ for which $\Omega(\sqrt{d}/\delta_T)$ is a lower bound.   
       
We analyze the greedy algorithm by proving the existence of a bounded convex
$K_T$ that is \emph{$T$-absorbing}: $\forall \vr x \in K_T$ and $\vr t \in \pm
T$, $\innerprod{\vr x}{\vr t} \leq 0 \Rightarrow \vr x+\vr t \in K_T$. We give
an explicit construction of a set $K_T$ contained in a ball of radius $(2/\delta_T)^{d-1}$, based on chains of subspaces spanned by vectors in $T$, which may be of
independent interest.  

We further generalize our greedy vector balancing bound to the setting of
online vector partitioning, where the sequence $\vr t_1, \dots, \vr t_n$ must
be partitioned in an online manner into $p$ subsequences of nearly equal sum.
As an application, we prove a special case of a conjecture of Bosman et al.
(Theory of Computing Systems, 2025), which implies that a lexicographic version
of total completion time scheduling under scenarios is polynomial time solvable
when the number of scenarios is fixed. 
%
%
\end{abstract}

\newpage

\section{Introduction} \label{sec:intro}

In online vector balancing, the task is to assign signs $s_1,s_2,\dots \in
\{\pm 1\}$ to an online sequence of vectors $\vr t_1, \vr t_2, \dots$ from a
universe $T \subseteq \R^d$, such that the signed combinations $\sum_{i=1}^k
s_i \vr t_i$, $k \geq 1$, have as small norm (also called discrepancy) as
possible in a given target norm. The online constraint is that the sign of a
vector must be chosen immediately after it arrives without knowledge of the
future. This problem has been studied extensively in the literature, across
adaptive adversary
models~\cite{spencer1977balancing,lagarias1977discrete,barany1979class,spencer1986balancing,doerr2001vector,barany2025balancing},
oblivious adversary
models~\cite{bansal2020online,alweiss2021discrepancy,liu2022gaussian,kulkarni2024optimal},
and stochastic models~\cite{aru2018balancing,bansal2021online} both in terms
existential bounds as well as algorithms. Algorithmic applications have been
given to provide improved SGD convergence rates~\cite{lu2022grab}, improved
algorithms for numerical integration~\cite{bansal2025quasi,dwivedi2024kernel},
as well as online item allocation~\cite{benade2018make}.    

The online vector balancing model was first proposed by
Spencer~\cite{spencer1977balancing,spencer1986balancing} under the framework of
a two player game, where the ``Pusher'' player chooses the sequence of vectors
in $T$ and the ``Chooser'' picks the signs. One of the simplest and most
natural online strategies for the Chooser, which will be the focus of this
work, is the Euclidean \emph{greedy vector balancing algorithm}: at each
iteration, it picks a sign that (locally) minimizes the Euclidean norm of the
next iterate. The algorithm is extremely simple: at iteration $k$, greedy picks
a sign $s_k \in \{\pm 1\}$ such that $s_k \vr t_k$ has non-positive inner
product with the previous iterate $\vr x_{k-1} := \sum_{i=1}^{k-1} s_i \vr t_i$, where $\pm 1$ are both valid
when the inner product is zero. The $k$th iterate is then given by $\vr x_k =
\vr x_{k-1} + s_k \vr t_k$. We shall say that a sequence $0 := \vr x_0,\vr
x_1,\dots$ is a \emph{greedy sequence for $T$} (or greedy $T$-sequence), if it
is consistent with the iterates produced by the greedy algorithm on some input
sequence from $T$. We denote the norm of a greedy $T$-sequence
$0:=\vr x_0,\vr x_1,\dots,\vr x_n$ by $\max_{k \in [n]} \|\vr x_i\|$. 

As observed by Spencer~\cite{spencer1977balancing}, the $k$th iterate $\vr x_k$
of a greedy sequence for the unit Euclidean sphere ${\cal S}^{d-1}$ in $\R^d$
has norm at most $\sqrt{k}$, as the squared norm increases by at most $1$ at
each iteration. This bound is also optimal for \emph{any online algorithm}, as
one may pick the input sequence such that each next vector is a unit vector
orthogonal to the previous iterate. We note that much better bounds are known
when the Pusher is \emph{oblivious} to choices made by Chooser instead of
adaptive, that is, when the Pusher must commit to the input sequence in
advance. In this case, randomized strategies allow for norm bounds growing
polylogarithmically with the sequence length
$n$~\cite{bansal2020online,alweiss2021discrepancy,liu2022gaussian,kulkarni2024optimal}. The focus of this work
will be an improved analysis of greedy sequence norms however, where there is
no difference between an adaptive and oblivious Pusher. 

While the above paragraph provides a complete understanding when the universe
$T$ is the unit sphere, there are multiple contexts appearing in applications
in which the universe $T$ is instead finite and where one may hope for stronger
greedy sequence norm estimates. 
Our primary application will be in the context of scenario
scheduling~\cite{bosman2025total}, which we detail later, where $T$ consists
of binary vectors indicating which scenarios each job participates in and where
the optimal solution induces a greedy $T$-sequence. We also mention an
application by Lu, Guo and De Sa~\cite{lu2022grab} in the context of stochastic
gradient descent (SGD), where $T$ consists of the possible sample gradients and
greedy is used by a sample reordering algorithm called Gradient Balancing
Algorithm (Grab).

\para{\bf Main Contribution} Motivated by the above, we ask the following
basic question: do greedy $T$-sequences have uniformly bounded norm,
independent of the length of the sequence, when $T$ is finite? In this context,
by choosing $T$ to be an $\epsilon$-net of the sphere, it is certainly
conceivable that one might be able to construct a greedy sequence whose norm
tends to infinity. As our main contribution, we in fact show that this is not
possible, and thus answer the above question in the affirmative. Somewhat
surprisingly, we are not aware of any prior work in this direction, even for
natural special cases such as $T = \{0,1\}^d$. 

Our results are quantitative and expose a connection to geometric parameter of
a point set on a sphere called the $\delta$-distance property. This parameter
has been extensively studied in the linear programming
literature~\cite{bonifas2012sub,brunsch2013finding,brunsch2014solving,dadush2016shadow,dadush2022finding,ekbatani2022circuit,dadush2026excluding},
where it has been shown to give a strong notion of numerical complexity of the
(normalized) rows of a constraint matrix. We give the definition below for
point sets on the unit sphere. 

\begin{definition} \label{def:delta}
A finite set of vectors $T \subset {\cal S}^{d-1}$ satisfies the $\delta$-distance
property if 
\[
\min \{\dist{\spn U,\vr t}: U \subset T, \vr t \in T \setminus \spn U\} \ \geq \ \delta.
\]
where $\spn{\cdot}$ is the linear span and $\dist{\cdot,\cdot}$ is the minimum
Euclidean distance. We further define $\delta_T \geq 0$ to be the largest $\delta \geq 0$ for which $T$ satisfies the $\delta$-distance property.
\end{definition}


Note that $\delta_T$ is strictly positive since each term in the minimum is
positive (since $\vr t \notin \spn U$). 
If $T$ is an infinite subset of
the sphere we always have $\delta_T = 0$ however, since by compactness of
${\cal S}^{d-1}$ there is an infinite sequence of distinct points $\vr
t_1,\vr t_2,\dots \in T$ satisfying $\lim_{i \rightarrow \infty} \norm{\vr
t_i-\vr t_{i+1}} = 0$. 

With this definition in hand, we may state our main result.

\begin{restatable}{theorem}{greedybalancing}\label{thm:main} 
Let $T \subseteq {\cal S}^{d-1}$ be a finite set. Define $G(T)$ to the maximum
norm of any greedy sequence for $[-1,1] T := \{a \vr t: a \in [-1,1], \vr t \in
T\}$. Then $G(T) \leq (2/\delta_T)^{d-1}$.  
\end{restatable}

Note that $G(T)$ in fact controls the norm of greedy sequences induced by $[-1,1]$ scalings of elements in $T$. Given this, we may use the above theorem 
to obtain estimates for arbitrary sets by appropriate scaling. Letting
$\ball{\vr 0}{R}$ denote the ball of radius $R$ in $\R^d$, we obtain the following
direct corollary:

\begin{corollary}\label{cor:scaling} 
For $T \subseteq \ball{\vr 0}{R}$, let $\hat{T} := \{\frac{\vr t}{\norm{\vr t}}: \vr
t \in T-\{0\}\}$. Then, we have that $G(T) \leq R G({\hat T}) \leq R (2/\delta_{\hat T})^{d-1}$. 
\end{corollary}

For the special case of $T = \{0,1\}^d$, it was shown by Alon and
Vu~\cite[Theorem 3.2.2]{ALON1997133} that the point to subspace distance $\min
\{\dist{\spn S, \vr t}: \vr t \notin \spn S, \vr t \cup S \subseteq \{0,1\}^d\}
= d^{-d/2+o(d)}$. Since every vector in $T=\{0,1\}^d$ has Euclidean norm at
most $\sqrt{d}$, by scaling this implies that $\delta_{\hat{T}} =
d^{-d/2+o(d)}$ as well. Applying the above corollary, we conclude that
$G({\{0,1\}^d}) \leq d^{d^2/2+o(d^2)}$. As stated previously, this is the first
finite estimate on the behavior of the greedy algorithm even for this special
case. We will see an application of the bound for $G(\{0,1\}^d)$ to scenario
scheduling in \Cref{sec:intro-scheduling}. 

We also briefly mention that \Cref{cor:scaling} can be used to bound an
important parameter required in the proof of convergence of Grab enhanced
SGD~\cite{lu2022grab}. Specifically, they require a finite bound on the norm of
greedy sequences induced by potential sample gradients~\cite[Assumption
5]{lu2022grab}, which \Cref{cor:scaling} justifies when the sample gradients are
bounded and take on only a finite number of directions. This captures,
for example, problems of the form $\min_{\vr x \in {\cal D}}
\frac{1}{n} \sum_{\vr y \in S} f(\vr x, \vr y)^p$, $p > 0$, where $S = \{\vr
y_1,\dots, \vr y_n\}$ is the sample space, ${\cal D} \subseteq \R^d$ is a
bounded domain, and the functions $f(\cdot,\vr y), \forall \vr y \in S$, are
continuous piecewise linear.

We do not know if the estimate in \Cref{thm:main} is tight. We are however able
to give a simple lower bound showing that a linear dependence on $1/\delta_T$
is necessary in the worst-case. The lower bound consists of orthogonal copies
of a two-dimensional greedy sequence for the set $[-1,1]T$ where $T$ is the set of nearly parallel vectors $\{\vr u_1, \vr v_1\} := \{(0,1),(\delta,-\sqrt{1-\delta^2})\}$ for $\delta \in (0,1)$. The two-dimensional case is illustrated in \cref{fig:lower_bound}. This yields the following result.

%
%

\begin{figure}[t!]
    \centering
\begin{tikzpicture}[
    >=Stealth,
    thick,
    scale=1
]

\coordinate (P0) at (0,0);

\coordinate (P1) at (0,2);
\coordinate (P2) at (1,0);
\coordinate (P3) at (1,2);
\coordinate (P4) at (2,0);
\coordinate (P5) at (2,2);
\coordinate (P6) at (1,4);
\draw[->, color=blue] (P0) -- (P1) node[midway,left] {$w$};
\draw[->] (P1) -- (P2) node[midway,left] {$v_1$};

\draw[->, color=blue] (P2) -- (P3) node[midway,left] {$w$};
\draw[->] (P3) -- (P4) node[midway,left] {$v_1$};

\draw[->, color=blue] (P4) -- (P5) node[midway,left] {$w$};

\draw[->, thick] (P5) -- (P6) node[midway,left] {$-v_1$};

\node[right] at (P4) {$(\frac{1}{\delta}, 0)$};

\end{tikzpicture}
    \caption{Example of a greedy sequence in which the difference between consecutive elements alternates between $w \coloneqq \sqrt{1-\delta^2} \cdot u_1$ and $v_1$ until the sequence reaches $\left(\tfrac{1}{\delta}, \sqrt{1-\delta^2}\right)$, at which point the final difference is $-v_1$.}
    \label{fig:lower_bound}
\end{figure}
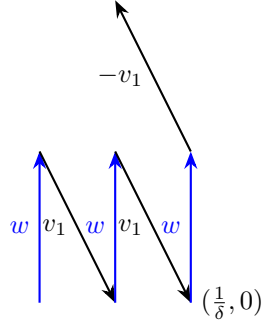

\begin{restatable}{lemma}{lowerbound}\label{lem:greedy-lower-bound}
For every $d \in \N_+$ and every $\delta \in (0,1)$ there is a set $T \subseteq {\cal S}^{2d-1}$ with $\delta_T = \delta$ and $ G(T) \geq \sqrt{d}/\delta$.
\end{restatable}

It is instructive to compare the bounds we achieve for greedy sequences to the best
uniform bounds achievable by any online algorithm. An essentially optimal
algorithm and bound in this regard has been given by
B{\'a}r{\'a}ny~\cite{barany1979class}. For a finite set $T \subset \R^d$ of
non-collinear vectors (pairs of distinct vectors in $T$ are linearly
independent) and target norm $\norm{\cdot}$, B{\'a}r{\'a}ny proved that the
discrepancy of the best online algorithm against sequences in $T$ is at most
$\max_{\vr y \in \{-1,1\}^T} \norm{\sum_{\vr t \in T} y_{\vr t} \vr t}$ and at
least $1/2$ this quantity. The upper bound, which in fact also holds for
sequences in $[-1,1] T$, is obtained by the following very simple online
algorithm: choose signs so as to maintain a representation of the iterates of
the form $\sum_{\vr t  \in T} y_{\vr t} \vr t$ with $\vr y \in [-1,1]^T$. By convexity, note that the worst-case norm is achieved when $y \in \{\pm 1\}^T$. Such a representation can be
maintained using the fact that if $y,\Delta y \in [-1,1]$ then either $y
- \Delta y$ or $y+\Delta y$ is also in $[-1,1]$. While simple and (perhaps
  surprisingly) optimal, it does not diminish the relevance of
understanding greedy sequences and their applications. Furthermore, maintaining
the requisite representation may become prohibitive if the universe $T$ is
large (i.e., $T = \{0,1\}^d$).  

Under the Euclidean norm, one can show that\footnote{Compute
${\mathbb E}_{\vr v \sim {\rm unif}({\cal S}^{d-1})}[\max_{\vr y \in \{-1,1\}^T}
\innerprod{\sum_{\vr t \in T} y_{\vr t} \vr t}{\vr v}]$ and then apply
Cauchy-Schwarz.} 
\[
\frac{\sqrt{2}}{\sqrt{\pi d}} \sum_{\vr t \in T} \norm{\vr t} \leq \max_{\vr y \in \{-1,1\}^T} \norm{\sum_{\vr
t \in T} y_{\vr t} \vr t} \leq \sum_{\vr t \in T} \norm{\vr t}.
\]
In particular, for $T \subseteq {\cal S}^{d-1}$, B{\'a}r{\'a}ny's bound under
the Euclidean norm is $\Theta_d(|T|)$, which is essentially independent of the
geometry of $T$. In contrast, bounds for greedy $T$-sequences must depend on
how ``well-separated'' the points in $T$ are on the sphere, which we encode
using the parameter $\delta_T$. An important similarity however is that both
the greedy and optimal bound are finite precisely when $T$ is finite. In the
concrete case of $T=\{0,1\}^d$, the optimal bound is $\Theta(\sqrt{d} 2^d)$
whereas our greedy upper bound is $d^{d^2/2+o(d^2)}$.

\para{\bf Overview of the Greedy Sequence Bound} We provide a detailed overview of our
approach to \Cref{thm:main}. The main idea is to a construct a bounded convex
set which contains all greedy $T$-sequences. To this end, we define a set $\KKK
\subseteq \R^d$ to be \emph{$T$-absorbing} if it contains the origin and if $\forall
\vr x \in \KKK$ and $\vr t \in T \cup -T$, $\innerprod{\vr x}{\vr t} \leq 0$
implies that $\vr x + \vr t \in \KKK$. It is direct to see by induction that
any such $K$ contains the iterates of all greedy $T$-sequences. If $K$ is also
convex, then this last statement automatically extends to $[-1,1]T$-sequences
as well, since then $\vr x \in K$ and $\vr x + \vr t \in K$ implies that $\vr x + [0,1] \vr t \subseteq K$. 

Our notion of $T$-absorbing above is in fact a specialization of
B{\'a}r{\'a}ny's notion of $T$-closedness~\cite{barany1979class}: $K$ is
$T$-closed if $\forall \vr x \in K$, $\forall \vr t \in T$, either $\vr x + \vr
t \in K$ or $\vr x - \vr t \in K$. B{\'a}r{\'a}ny associates every
online algorithm to a $T$-closed set containing the origin, which our
definition specializes to the greedy algorithm (note $T$-absorbing
$\Rightarrow$ $T$-closed). In the case that $T$ consists of non-collinear
vectors, B{\'a}r{\'a}ny approximately characterized the smallest $T$-closed
sets by showing that the convex hull of any $T$-closed set contains a shift of
the ``universal'' $T$-closed set $\{\sum_{\vr t \in T} y_t \vr t: \vr y \in
\{0,1\}^T\}$ (shifting does not affect $T$-closedness). For the
$T$-absorbing case, it is formally easier to define the smallest such set as
the intersection of all (convex) $T$-absorbing sets is again (convex)
$T$-absorbing. Note that these intersections are non-empty since $\R^d$ is
trivially $T$-absorbing and all $T$-absorbing sets contain the origin. Our
construction will crucially make use of some of the structural properties of
the smallest such sets, though they seem more difficult to
characterize (even approximately) than in the $T$-closed setting. 

A first property we shall exploit in our construction is that the smallest (convex)
$T$-absorbing set $\KKK_T$ is always origin symmetric, i.e., $\KKK_T =
-\KKK_T$. This follows since $-\KKK_T$ is also $T$-absorbing and thus $\KKK_T =
\KKK_T \cap -\KKK_T$ by minimality. From here, we restrict ourselves to convex $T$-absorbing
sets, which can be characterized in terms of their projections. Let $\proj{\vr
t}(\cdot)$ denote the orthogonal projection to the subspace orthogonal to $\vr
t$ and let $\interval {-\vr t}{\vr t}$ the line segment from $-\vr t$ to $\vr
t$. We show that for a convex set $\KKK$ containing the origin, $\KKK$ is
$T$-absorbing iff it is \emph{$T$-projective}, which we define as:
\begin{equation}
\proj{\vr t}(\vr x) + \interval{-\vr t}{\vr t} \subseteq \KKK, \forall \vr t \in T, \vr x \in \KKK, \label{def:projective-set}
\end{equation}
Usefully, if $\KKK$ is $T$-projective then so is its convex
hull, which we denote by $\convhull(\KKK)$. This implies that one only needs to
test the left hand side condition when $\vr x$ is a vertex of $\KKK$, which
helps simplify the analysis of our eventual construction.  

Using this definition, when $T \subset {\cal S}^{d-1}$ consists of unit
vectors, we give a recursive construction of an origin symmetric convex set $\KKK_T \subset
\spn{T}$ of radius bounded by $R_{\dim{T}} := (2/\delta_T)^{\dim{\spn T}-1}$.
The construction proceeds by induction on dimension. If $T$ has dimension $1$, then
$T = \{\pm \vr t\}$ and clearly the smallest $T$-projective set is $\KKK_T :=
\interval{-\vr t}{\vr t}$, which has radius $R_1 = 1$. 

\begin{figure}[b!] 
\centering
\includegraphics[width=5cm]{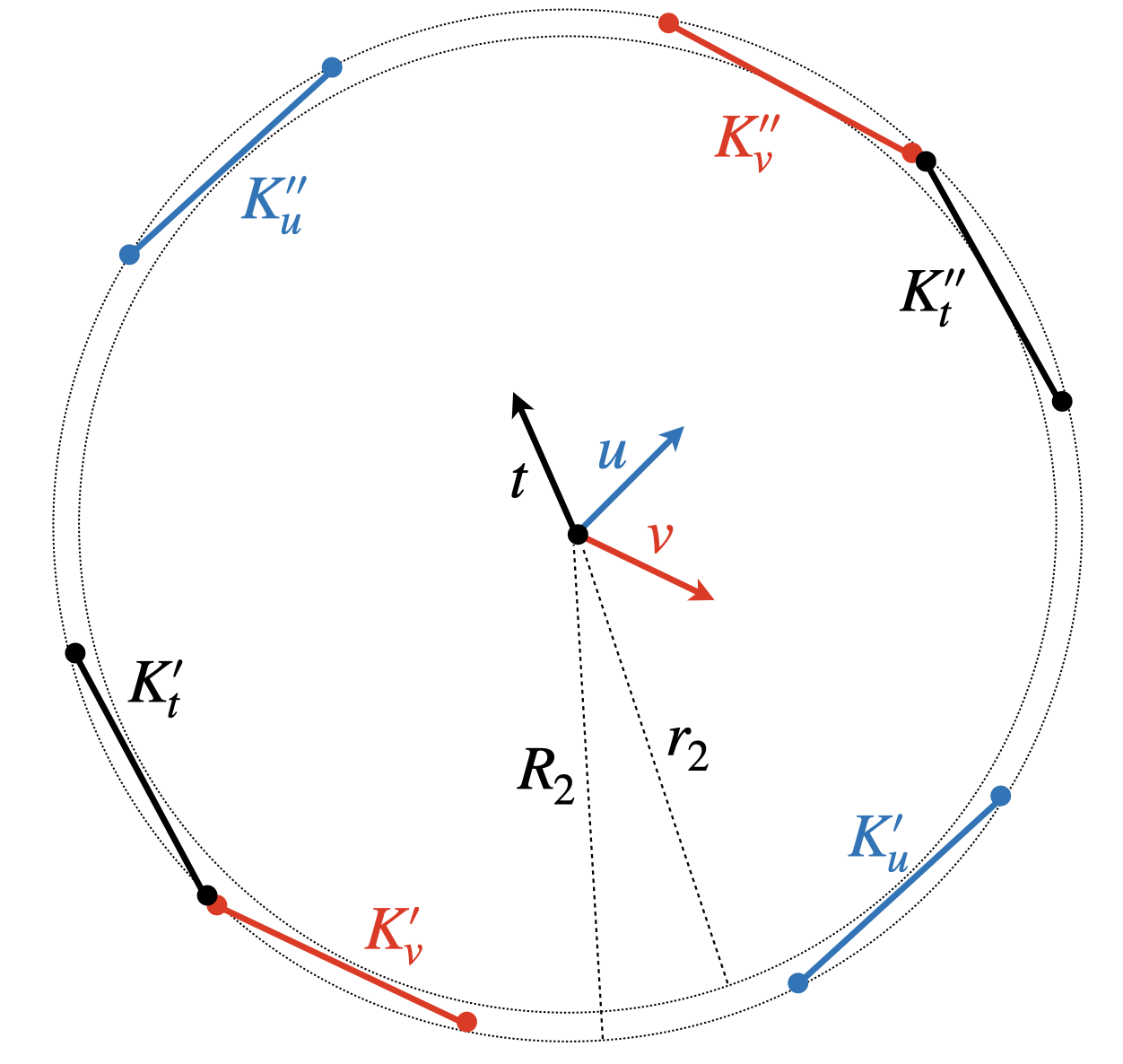}
\caption{Let $d=2$ and $T = \{\vr t, \vr u, \vr v, -\vr t, -\vr u, -\vr v\}$.
Our construction for $\KKK_T$ is the convex hull of 6 line segments:
$K'_{\vr t}, K''_{\vr t}, K'_{\vr u}, K''_{\vr u}, K'_{\vr v}, K''_{\vr v}$.}
\label{fig:K2}
\end{figure}

In the two dimensional case, the construction is again simple. We now derive
from first principles what the smallest $T$-projective convex set
$\KKK_T$ must look like in two dimensions, and then give a slightly larger
explicit construction. Let $\hypp {\vr t}$ denote the subspace orthogonal to
$\vr t \in T$. When $T$ is two dimensional, the restriction of $\hypp {\vr t}$
to $\spn{T}$ can be expressed as the span of a single unit vector, which we
denote by $\vr r_{\vr t}$. Consequently, the projection $\proj{\vr t}(\KKK_T)$
is then a symmetric interval of the form $\ell_{\vr t} \interval{-\vr r_{\vr t}}{\vr
r_{\vr t}}$ for some choice of length $\ell_t \geq 0$. The set of points
appearing on the left hand side of \eqref{def:projective-set} for a fixed $\vr
t \in T$ is the origin symmetric rectangle $\proj{\vr
t}(\KKK_T)+\interval{-{\vr t}}{{\vr t}} = \ell_{\vr t} \interval{-\vr r_{\vr
t}}{\vr r_{\vr t}} + \interval{-\vr t}{\vr t}$. We express this rectangle as the
convex hull of its edges in the direction of $\vr t$, which we denote by
$K_{\vr t}^{'} := \ell_{\vr t} \vr r_{\vr t} + \interval{-\vr t}{\vr t}$ and
$K_{\vr t}^{''} := -\ell_{\vr t} \vr r_{\vr t} + \interval{-\vr t}{\vr t}$.
Importantly, if $\KKK_T$ is $T$-projective, then so is
$\convhull\left(\bigcup_{\vr t \in T} \proj{\vr t}(\KKK_T) + \interval{-\vr
t}{\vr t}\right)$ as condition \eqref{def:projective-set} gets easier to
satisfy. Therefore, the smallest $\KKK_T$ must be the convex hull these
rectangles, and hence the only degree of freedom is choice of lengths $\ell_{\vr t}$,
$\vr t \in T$. For our construction, we use the uniform choice $\ell_{\vr t} =
r_2$, $\forall \vr t$. With this choice, condition~\eqref{def:projective-set}
is satisfied when $r_2$ is large enough to ensure that the edges $\{K_{\vr
t}^{'}, K_{\vr t}^{''}: \vr t \in T\}$ are all parwise interior disjoint. See
figure~\ref{fig:K2} for an illustration. Noting that in two dimensions
$\delta_T = \sin(\alpha_T)$, where $\alpha_T \in (0,\pi/2]$ is the minimum
non-zero angle between vectors in $\pm T$, it is not hard to check that setting
$r_2 = \cos(\alpha_T/2)/\sin(\alpha_T/2)$ is the minimum value for this
purpose. The final radius is then $R_2 = \sqrt{R_1^2+r_2^2} =
1/\sin(\alpha_T/2) \leq 2/\sin(\alpha_T) = 2/\delta_T$. 

Extending the above construction to $\dim{T} \geq 3$ is unfortunately more
complicated. Inspired by the above, the main idea is still to ``orthogonally
extend'' absorbing sets built from lower dimensional subsets of $T$ and take
convex hull. In two dimensions, we took each $1$-dimensional absorbing sets
$\interval{-\vr t}{\vr t}$ and added the orthogonal interval $r_2
\interval{-\vr r_{\vr t}}{\vr r_{\vr t}}$ to it. In higher dimensions, we
will examine subsets $U \subseteq T$ which are \emph{maximal} subject to
dimension, meaning $T \cap \spn{U} = U$, and which have strictly smaller
dimension than $T$. We denote this collection of subsets by ${\cal U}_T$. By
the induction hypothesis, we have already constructed a $U$-absorbing set
$\KKK_U \subset \spn{U}$ of radius $R_{\dim{U}}$. It is tempting at this point
to restrict attention only to sets $U \in {\cal U}_T$ with dimension precisely
$\dim{T}-1 := d-1$. In this case, the orthogonal complement $\spn{U}^\perp$
restricted to $T$ is again one dimensional, spanned a unit vector $\vr r_{U}$.
The direct extension of the two dimensional construction would then be to
define $\KKK_T := \convhull\left(\bigcup_{U \in {\cal U}_T, \dim{U}=d-1}
\KKK_U + r_{d} \interval{-\vr r_U}{\vr r_U}\right)$ for the appropriate
choice of $r_{d}$. Unfortunately, with this construction, it is unclear
why $\KKK_T$ is $T$-projective for any choice of $r_{d}$. For a pair
$(\vr t,U)$, where $\vr t \notin U$, it is not even clear why 
$\proj{\vr t}(\KKK_U) \subseteq \KKK_T$, 
since the main thing we know about $\KKK_U$ is the bound of $R_{d-1}$ on its radius.

To remedy this issue, we will need to orthogonally extend sets $\KKK_U$ of all
dimensions in the construction of $\KKK_T$, where lower dimensional sets
$\KKK_U$ are extended in more directions. Specifically, to each set $\KKK_U$,
$U \in {\cal U}_T$, we add the entire ball $\ball{\vr 0}{r_{d}}$ restricted to
the orthogonal subspace $\spn{U}^\perp \cap \spn{T}$. We then take the convex
hull of all these sets to construct $\KKK_T$. Note that choosing $r_d \geq R_{d-1}$, 
we at least resolve the last issue in the above paragraph since
$\proj{\vr t}(\KKK_U) 
\subseteq 
\proj{\vr t}(\ball{\vr 0}{R_{d-1}} \cap \spn{U} 
)
\subseteq
\ball{\vr 0}{r_d} \cap \hypp{t} \cap \spn{T} \subseteq \KKK_T$. 
A similar
choice of $r_d = R_{d-1} \cos(\alpha_T/2)/\sin(\alpha_T/2)$, where $\alpha_T$
is now the minimum non-zero angle between vectors and subspaces spanned by $T$,
can now be shown to work, however the proof is more complicated. Interestingly,
unfolding the construction, one can precisely characterize the basic building blocks
whose convex hull form $\KKK_T$, which are each indexed by a chain of subspaces
spanned by $T$. We defer the additional details to \Cref{sec:result}.        

\subsection{Generalization to Vector Partitioning} 
We now consider the following
partition generalization of the online vector balancing game, which our results naturally extend to. Let a finite
vector set $T\subseteq \mathbb R^d$ be given together with a number $p \geq 2$
of desired partition pieces. The game between Pusher and Chooser starts with
the trivial partition $J_1\dot\cup\ldots \dot \cup J_p$ with $J_i=\emptyset$
for all $i\in \{1,\ldots ,p\}$. Pusher picks a vector $\mathbf t_k\in T$ at
time $k$, and Chooser has to assign the vector $\mathbf t_k$ to one of the
partition pieces by adding the index $k$ to one of the sets $J_i$, that is, the
set $J_i$ is replaced by the set $J_i\dot \cup \{k\}$. Pusher seeks to maximize
and the Chooser minimize the largest difference between the sums of vectors
indexed by pieces of the partition. Letting $[p] := \{1,\dots,p\}$, the value of the partition $J_1 \dot \cup \dots \dot \cup J_p$ is then defined as
\[
\mathrm{val}(J_1,\ldots, J_p):= \max_{n \geq 1} \max_{i_1 \in [p]}\max_{i_2 \in [p]} \left\|\sum_{j\in J_{i_1} \cap [n]}\mathbf t_j-\sum_{j\in J_{i_2} \cap [n]}\mathbf t_j\right\|.
\]

It is straightforward to observe that the case $p=2$ is equivalent to the original vector balancing game
. Indeed, if $J_1$ resp.~$J_2$ are the set of vectors to which Chooser assigns a plus sign $s_j=1$ resp.~a minus sign $s_j=-1$, then the value of the game after the $n$-th vector is
equal to $   
   \left\|
   \sum_{j\in J_1}\mathbf t_j-
   \sum_{j\in J_2}\mathbf t_j\right\|=
   \left\|\sum_{j=1}^ns_j\mathbf t_j\right\|$.

The vector partitioning problem (in the offline setting) was first introduced by B{\'a}r{\'a}ny and
Doerr~\cite{barany2006balanced,barany2008power}. They proved that one
can always construct a partition $J_1\dot\cup\ldots \dot \cup J_p$ 
such that $\mathrm{val}(J_1,\ldots, J_p) = O(d \max_{\vr t \in T}
\norm{\vr t})$, for any norm $\norm{\cdot}$. An extension of the greedy
algorithm to vector partitioning was analyzed by Aru, Narayanan, Scott, and
Venkatesan~\cite{aru2018balancing} in the stochastic setting, where they showed
that greedy achieves asymptotically near-optimal norm bounds when the vectors $\vr
t_1,\vr t_2,\dots$ are distributed iid from a ``nice'' distribution on the unit
ball. In this setting, the greedy algorithm, which they dubbed the \emph{inner
product rule}, proceeds as follows: at step $k$ assign $\vr t_k$ to a partition
piece with which it has the smallest inner product, i.e., to any piece in ${\rm
argmin}_{i \in [p]} \innerprod{\sum_{j \in J_i} \vr t_j}{\vr t_k}$.
We call any partition consistent with a run of the greedy algorithm a \emph{greedy
partition}. One can easily verify that for a greedy partition $J_1 \dot \cup J_2$
when $p=2$, assigning $+1$ to the elements $J_1$ and $-1$ to those $J_2$
yields a signing consistent with a run of the greedy vector balancing
algorithm.

In the online context, our main result is that greedy's value for partitioning
is at most its vector balancing value. That is, the $p=2$ case is the hardest case
of the game. 

\begin{restatable}{theorem}{greedypartition}\label{thm:greedy-partition}
Let $p \geq 2$, $\vr t_1, \dots, \vr t_n$ be an input sequence from $T
\subseteq \R^d$. Then, any greedy partition $J_1 \dot \cup \dots \dot \cup
J_p = [n]$ for $\vr t_1,\dots, \vr t_n$ satisfies $\mathrm{val}(J_1,\dots,J_p) \leq G(T)$, where $G(T)$ is the greedy
vector balancing bound as defined in \Cref{thm:main}.   
\end{restatable}

The result follows from the simple but powerful observation that any two pieces
$J_k,J_l$ of a greedy $p$-partition for $\vr t_1,\dots,\vr t_n$ also form a
greedy $2$-partition of the subsequence induced by $J_k \cup J_l$. That is, the
greedy partitioning algorithm is consistent with a run of the greedy vector
balancing on any two of its pieces. The bound therefore follows blackbox from
the vector balancing case. We note that B{\'a}r{\'a}ny's online vector
balancing upper bound also extends to partitioning via a natural round-robin
extension of his balancing algorithm.  

\subsection{Application to Scenario Scheduling} 
\label{sec:intro-scheduling}

We now detail an application of the greedy vector partitioning bound to the
problem of total completion time scheduling under scenarios introduced by
Bosman, van Ee, Ergen, Imreh, Marchetti-Spaccamela, Skutella and
Stougie~\cite{bosman2025total}. In this problem, we are given $n$ weighted jobs
with unit processing times to be scheduled on $p$ identical machines, together
with $d$ scenarios $S_i \subseteq [n]$, $i \in [d]$, 
each of them represented by the subset
of jobs that are present in that scenario. A global assignment of jobs to
machines is adapted to any given scenario simply by restricting it to the jobs
present. The goal is to compute a global assignment of jobs to machines that
minimizes the average sum of weighted completion times over all jobs, where the
average is taken over scenarios. The formal description is given below: 

\begin{center}
\noindent\fbox{%
    \parbox{0.95\textwidth}{%
     \textsc{  MinAvgSTC($d$)}
   
       \textit{Input:} Numbers $n,m\in \mathbb{N}$ of jobs resp.~machines, weights $w_j$ for $j\in [n]$ that are non-increasing in $j$, scenarios $S_k\subseteq [n]$ for $k\in [d]$.

       \textit{Task:} Find a partition $[n]=J_1\dot\cup\ldots\dot\cup J_m$ of jobs so as to minimize 
        \begin{align}\label{eq:schedobj}
           \mathrm{obj}(J_1,\ldots,J_m)\coloneqq \sum_{k=1}^d \sum_{i=1}^m \sum_{j\in J_i\cap S_k}w_j\cdot|J_i\cap S_k\cap [j]\}|.
        \end{align}
    }%
}    
\end{center}

When there is only one scenario, it is classical that scheduling the jobs in a
round robin fashion in non-increasing order of weight is
optimal~\cite{eastman1964bounds}. \cite{bosman2025total} showed that problem
is NP-hard when the number of scenarios is part of the input even when $m=2$.
However, they conjectured that the problem is polynomial for any fixed number
of scenarios. Furthermore, they provided a polynomial time dynamic programming
algorithm that solves the problem under the conjectured existence of a
\emph{balanced optimal} solution, which we state below:

\begin{conjecture} [\cite{bosman2025total}]
  \label{conj:minsum}
\textsc{MinAvgSTC$(d)$} has an optimal solution $(J_1,\ldots, J_m)$ such that for every scenario $k\in [d]$ and each $j\in [n]$, the $j$ largest jobs are assigned to the machines in such a way that the difference in number of jobs assigned to each pair of machines is bounded by a function $g(d)$ of $d$ only, or more formally 
\[ 
\max_{j\in [n], k\in [d]} \left\{  \max_{i\in [m] } |J_i\cap S_k\cap [j]| - \min_{i \in [m]} |J_i\cap S_k\cap [j]| \right\} \leq g(d). 
\] 
\end{conjecture}

\cite{bosman2025total} proved the conjecture for the case of unit weights
$w_1=\cdots=w_n=1$. Using the greedy partitioning, we
prove the conjecture in the complementary regime of exponentially decreasing
weights:

\begin{restatable}{theorem}{schedulingconj}\label{thm:scheduling-conj}
Conjecture \ref{conj:minsum} holds in the special case where $w_j\geq 2dw_{j+1}$ for all $j\in [n-1]$.  
\end{restatable}

The theorem implies that the dynamic program proposed by Bosman et
al.~\cite{bosman2025total} can compute an optimal solution for the above choice of
weights with a running time of $O(nm^{O(g(d))})$, which is polynomial when the
number $d$ of scenarios is fixed.

For the weights as above, we in fact show that \textsc{MinAvgSTC($d$)} becomes
a lexicographic minimization problem. Namely, the optimal solution minimizes
the average completion time of the largest weight job, and subject to this,
minimizes the average completion time of the second largest weight job, etc. With this interpretation in mind, one can prove that the partition $J_1
\cdot \cup \dots \cdot \cup J_p = [n]$ produced by the optimal solution is a
greedy partition with respect to the vectors $\vr t_1,\dots, \vr t_n$, where
$\vr t_i \in \{0,1\}^d$ is the indicator vector of the scenarios that job $j$
participates in. The greedy vector balancing bound for $T=\{0,1\}^d$ then
directly implies the conjecture in this case. We note that we still require the
dynamic program to compute the optimal solution in this setting as there can be
exponentially many non-equivalent greedy partitions.

\subsection{Organization}

In \Cref{sec:result}, we provide the full proof of \Cref{thm:main}, our greedy
vector balancing bound. In \Cref{sec:partition}, we prove the greedy
partitioning bound from \Cref{thm:greedy-partition}. In \Cref{sec:scheduling},
we prove \Cref{thm:scheduling-conj}, corresponding to the case of
\Cref{conj:minsum} for exponentially decreasing weights. Lastly, in
\Cref{sec:lower-bound}, we prove the simple lower on greedy vector balancing in
terms of the parameter $\delta_T$, corresponding to the proof of
\Cref{lem:greedy-lower-bound}.

\section{Greedy Balancing Bound}  \label{sec:result}

In this section we prove our main result, \Cref{thm:main}, stated in
\Cref{sec:intro}:
\greedybalancing*


\para{Proviso}
Throughout this section let $T \subset \R^m$ be a finite nonempty set of vectors
satisfying
$-T = \setof{-\vr t}{\vr t\in T} \subseteq T$.
We assume that
all vectors $\vr t\in T$ satisfy $\norm{\vr t}= 1$, and
that $\delta\in (0,1]$ is a constant such that
$T$ satisfies the $\delta$-distance property.
We conveniently use another constant $\eps := \sin(\alpha/2)$ where
$\delta = \sin(\alpha)$ (see Fig.~\ref{fig:triangle}, left).
Thus $\alpha$ 
is the minimal possible \emph{angle} between
$\vr t\in T$ and $\spn U$, for $\vr t \in T\setminus \spn U$ and $U\subset T$.
Let $d=\dim{\spn T}$. 
We do not assume $d=m$ as the proof will induct through subsets of a given $T\subset \R^m$.
Finally, let $R_i := (1/\eps)^{i-1}$ for $i\in [d]$.



\begin{figure}[t!] 
\centering
\begin{minipage}{.4\linewidth}
\includegraphics[width=5cm]{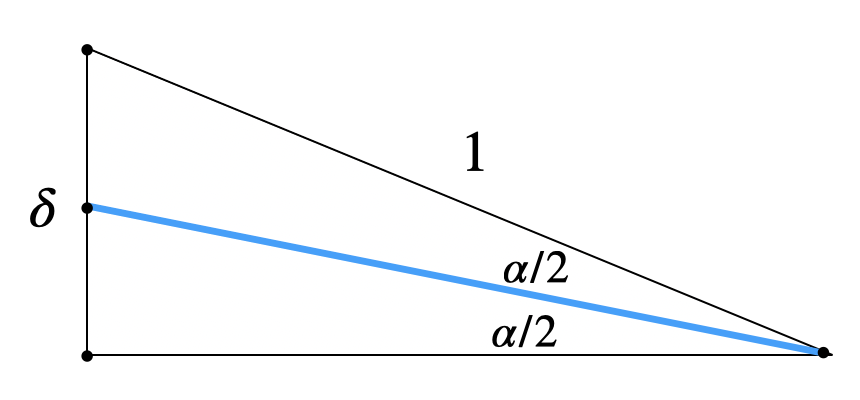}
\end{minipage}
\begin{minipage}{.4\linewidth}
\vspace{1cm}
\includegraphics[width=5cm]{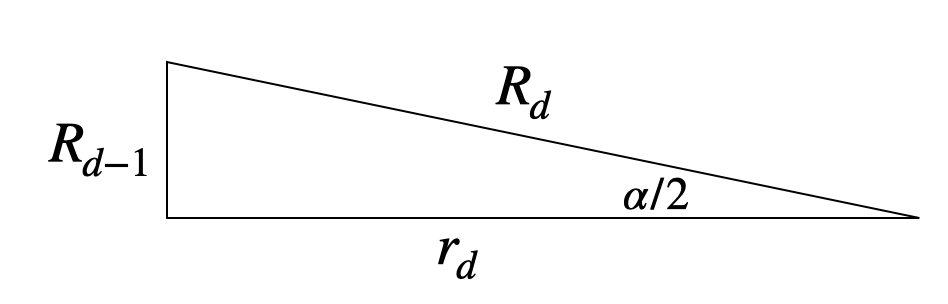}
\end{minipage}
\caption{Left: def. of $\eps := \sin(\alpha/2)$.  
Right: radiuses of balls in dimension $d-1$ and $d$.}
\label{fig:triangle}
\end{figure}

%



\medskip

We write $\hypp U$ for the subspace of $\R^m$ orthogonal to the set $U$ of vectors.
A special case of $\hypp {\{\vr t\}}$ we write as $\hypp {\vr t}$.
%
We denote by $\proj{H} : \R^m \to H$ the projection onto a subspace 
$H\subseteq \R^m$.
As a special case, the projection along a non-zero vector $\vr t \in\R^m$,
that is, the projection $\R^m \to \hypp {\vr t}$ onto the \emph{hyperplane} $\hypp {\vr t}\subseteq\R^m$
orthogonal to $\vr t$, is denoted as $\proj{\vr t}$ instead of $\proj{\hypp {\vr t}}$.
We write $\slball {\vr x} r$ and $\slsphere {\vr x} r$ to denote the closed ball
of radius $r$, 
and the corresponding sphere:
$\slball {\vr x}{r}  := \setof{\vr y\in \R^m}{\norm{\vr y - \vr x}\leq r}, 
\ 
\slsphere {\vr x}{r}  := \setof{\vr y\in \R^m}{\norm{\vr y - \vr x}= r}.$

\begin{definition} \rm
A convex subset $\KKK \subseteq \R^m$ satisfying $\vr 0 \in \KKK$ is 
\begin{itemize}
\item
\emph{$T$-absorbing} if for every vector 
$\vr x \in \KKK$
and every $\vr t \in T$ such that the inner product 
$\innerprod{\vr x}{\vr t}\leq 0$, we have
$\vr x + \vr t \in \KKK$;
\item 
\emph{$T$-projective} if for every $\vr t \in T$ the set 
$\proj{\vr t}(\KKK) + \interval {-\vr t}{\vr t}$ is included in $\KKK$,
where $\interval {-\vr t}{\vr t}$ is the closed line segment with endpoints 
$-\vr t$ and $\vr t$.
\end{itemize}
\end{definition}

The two notions are equivalent for convex sets, but we need one of implications only:
\begin{lemma} \label{lem:+-t}
Every $T$-projective convex set $\KKK\subseteq \R^m$ is $T$-absorbing.
\end{lemma}
\begin{proof}
Take any $\vr t \in T$
and $\vr x\in \KKK$ such that $\innerprod{\vr x}{\vr t}\leq 0$.
The point $\vr x + \vr t$ belongs to the line $\ell$ orthogonal to $\hypp {\vr t}$ 
passing through $\proj{\vr t}(\vr x)$.
Decompose $\vr x$ as $\vr x= \proj{\vr t}(\vr x) + \vr s$, where $\vr s$ is orthogonal
to $\hypp{\vr t}$, to deduce that $\innerprod{\vr s}{\vr t} \leq 0$.
Thus $\vr x + \vr t$ is on $\ell$
between points $\vr x$ and $\proj{\vr t}(\vr x) + \vr t$,
the latter point also belonging to $\KKK$ by assumption.
The point $\vr x + \vr t$ is thus a convex combination of two elements of $\KKK$,
and hence belongs to $\KKK$, due to its convexity.
\end{proof}

We derive \Cref{thm:main} immediately using the following result (combined with
\Cref{lem:+-t}):
\begin{theorem} \label{thm:abs}
For every $T\subset \R^m$ and $\delta > 0$ as in Proviso,
there is a $T$-projective set $\KK T\subseteq\spn T$ included in the ball 
$\slball {\vr 0} {R_d}$.
\end{theorem}

Indeed, 
the $T$-absorbing set $\KK T$ is also $\interval{-1}{1}T$-absorbing,
and a greedy sequence for $\interval{-1}{1}T$ cannot escape from a $T$-absorbing
set $\KKK$, which implies that its norm is bounded by 
$R_d = (1/\eps)^{d-1} \leq (2/\delta)^{d-1}$ since 
$
\delta = \sin(\alpha) = 2\sin(\alpha/2)\cos(\alpha/2) =   
2\eps\sqrt{1-(\eps)^2}  \leq  2\eps.
$
It thus remains to prove \Cref{thm:abs}, which occupies 
the rest of this section.

%
%

%

%
%
%

\subsection{Construction of sets $\KK T$}  \label{sec:construction}

Given $T\subset \R^m$ and $\delta > 0$,
%
we define the set $\KK T$
inscribed, intuitively speaking, in $\slsphere{\vr 0}{R_d}\cap\spn T$,
and circumscribed on $\slsphere{\vr 0}{r_d}\cap\spn T$,
where (see Fig.~\ref{fig:triangle} (right) for an illustration):
\[
r_1 \ := \ 1
\qquad\qquad\qquad
(r_d)^2 \ := \  (R_d)^2 - (R_{d-1})^2 \qquad \text{ for } d\geq 2.
\]
%
%
The construction of $\KK T \subseteq \slsphere{\vr 0}{R_d}\cap\spn T$ is inductive on 
$d = \dim{\spn T}$.

\para{Dimension $d=1$} 
$T = \{\vr t, -\vr t\}$, and $\norm{\vr t}=1$.
We define $\KK {\{-\vr t, \vr t\}} := \interval {-\vr t}{\vr t}$
(we write $\KK {\vr t}$ instead of $\KK {\{\vr t, -\vr t\}}$ below).

\para{Dimension $d=2$ (warmup)} 
Before proceeding with the general induction step, as a warmup we consider 
separately the case of the smallest dimension 2,
to settle geometric intuitions
(see Fig.~\ref{fig:K2} in the introduction).
By the previous paragraph, for every vector $\vr t\in T$ we have a $\{\vr t, -\vr t\}$-absorbing 
line segment
$\KK {\vr t}$ included in the line $\spn{\vr t}$.
The intersection $\hypp {\vr t} \cap \spn T$ is a line containing all vectors in $\spn T$ orthogonal to $\vr t$.
Let $\vr r_{\vr t}\in \hypp {\vr t} \cap \spn T$ be a vector orthogonal to $\vr t$ and 
of norm  
$
\norm{\vr r_{\vr t}} \ = \ r_2 \ = \ \sqrt{(1/\eps)^2 - 1}.
$
Thus $\vr r_{\vr t} \in \hypp {\vr t} \cap \spn T \cap \slsphere{\vr 0}{r_2}$.
There are two such vectors, $\vr r_{\vr t}$ and $-\vr r_{\vr t}$,
and we define two 'shifted copies' of $\KK {\vr t}$:
\begin{align} \label{eq:shifted}
K'_{\vr t} \ := \  \vr r_{\vr t} + \KK {\vr t}
\qquad\qquad
K''_{\vr t} \ := \ - \vr r_{\vr t} + \KK {\vr t},
\end{align}
where the addition stands for the Minkowski sum, namely
$\vr r_{\vr t} + \KK {\vr t}  = \setof{\vr r_{\vr t} + \vr v}{\vr v \in \KK {\vr t}}$ and
$- \vr r_{\vr t} + \KK {\vr t} = \setof{- \vr r_{\vr t} + \vr v}{\vr v \in \KK {\vr t}}$.
The sets \eqref{eq:shifted} are two copies of
the inverval $\KK {\vr t}$, shifted by the vector $\vr r_{\vr t}$ and $-\vr r_{\vr t}$, 
respectively, orthogonally to $\KK {\vr t}$.
We define the set $\KK T$ as the convex hull of the union of
all the line segments \eqref{eq:shifted} for all $\vr t\in T$:
\[
\KK T \ := \ \convhull \Big(\bigcup_{\vr t \in T} K'_{\vr t} \cup K''_{\vr t}\Big).
\]
\vspace{-4mm}

\para{Inductive step: $d\geq 2$} 
A nonempty subset $U\subset T$ satisfying $\dim{\spn U} < d$ is called \emph{maximal} if 
there is no set $U'$
such that $U\subset U' \subseteq T$ and $\spn U = \spn{U'}$.
The set $U$ is thus inclusion-maximal among all subset spanning the same vector subspace.
Let ${\cal U}_T$ be the set of all maximal subsets $U\subset T$.
By the induction assumption, for every $U\in {\cal U}_T$ we have an $U$-absorbing set
$\KK U$ included in $\spn U$.
The idea is to attach a copy of $\KK U$
to every point of the sphere $\slsphere{\vr 0}{r_d}\cap\spn T$ where the radius vector
is orthogonal to $\spn U$, for every $U\in{\cal U}_T$, 
and the to define $\KK T$ as the convex hull of all these 'shifted copies' of sets $\KK U$.

For every $U \in {\cal U}_T$, let $R_{U} := \hypp {U} \cap \spn T \cap \slsphere{\vr 0}{r_d}$
be the set of all vectors of norm $r_d$ orthogonal to $\spn U$.
We define the union of all 'shifted copies' of $\KK {U}$ attached tangentially to all points in $R_U$:
\begin{align} \label{eq:L}
L_{U} \ := \ R_{U} + \KK {U},
\end{align}
where the addition stands for the Minkowski sum, as usual.
In one extremal case, when $\dim{\spn U} = 1$ we have $U = \{-\vr t, \vr t\}$,
$\KK U =  \interval{-\vr t}{\vr t}$, and $L_U$ is a 'cylinder', namely
\begin{align} \label{eq:cylinder}
L_{\vr t} \ := \ L_{\{-\vr t, \vr t\}} = \big(\hypp {\vr t} \cap \spn T\cap \slsphere{\vr 0}{r_{d}}\big) 
+ \interval{-\vr t} {\vr t}.
\end{align}
In the other extremal case, when $\dim{\spn U} = d-1$,  we have
$L_U = \{-\vr r_{U}, \vr r_U\} + \KK {U}$, where $-\vr r_U, \vr r_U$ are the two
vectors orthogonal to $\spn U$ or norm $r_d$.
When dimension $d=\dim{\spn T}=2$,  both the extremal cases coincide.
%
We define the set $\KK T$ as the convex hull of the union of all the sets \eqref{eq:L}:
\[
\KK T \ := \ \convhull \Big(\bigcup_{U \in {\cal U}_T} L_{U}\Big).
\]
\vspace{-2mm}
%
%
\begin{claim} \label{lem:B}
$\KK T \subseteq \slball{\vr 0}{R_d}\cap\spn T$.
\end{claim}
\begin{claimproof}
By induction on $d$.
If $d=1$, 
$\KK {\vr t} := \interval {-\vr t}{\vr t}$
is equal to the 1-dimensional ball $\slball{\vr 0}{R_1}\cap\spn {\vr t}$.
If  $d=2$ (warmup), 
by construction, each of line segments $K'_{\vr t}, K''_{\vr t}$ is inscribed in the circle
$\slsphere {\vr 0}{R_2}\cap \spn T$,
where $R_2 = {1/\eps}$, and therefore $\KK T\subseteq \slball {\vr 0}{R_2} \cap \spn T$.

Suppose now that $d\geq 2$.
By induction assumption, for every $U\in {\cal U}_T$  we have
$\KK U \subseteq \slball {\vr 0} {R_{e}} \cap \spn U$, where $e = \dim{\spn U}$,
 which implies: 
\[
L_U \ \subseteq \ R_U + \big(\slball {\vr 0} {R_{e}} \cap \spn U\big) \ \subseteq \ 
R_U + \big(\slball {\vr 0} {R_{d-1}}\cap \spn U\big). 
\]
By the triangle of Fig.~\ref{fig:triangle} (right), the corresponding spheres 
are included in $\slsphere {\vr 0} {R_{d}}$:
\[
R_U + \big(\slsphere {\vr 0} {R_{d-1}}\cap \spn U\big) \ \subseteq \ \slsphere {\vr 0} {R_{d}} \cap \spn T,
\]
which implies the analogous inclusion 
$
R_U + \big(\slball {\vr 0} {R_{d-1}}\cap \spn U\big) \ \subseteq \ \slball {\vr 0} {R_{d}} \cap \spn T
$ of balls,
and therefore 
$
\KK T \ \subseteq \ \slball {\vr 0}{R_d} \cap \spn T.
$
\end{claimproof}

\para{Equivalent definition}
The inductive definition of $\KK T$ can be equivalently presented as 
the convex hull of multiple copies of just one-dimensional sets 
$\KK {\vr t} = \interval{-\vr t}{\vr t}$, 
for $\vr t\in T$,   
shifted in multiple orthogonal directions along vectors
of norm $r_i$, for $i \in [d]$.
Going even further, one can take instead of each $\KK{\vr t}$ only its endpoints 
$\{-\vr t, \vr t\}$.
This intuition leads to the following non-inductive but equivalent definition of $\KK T$.

Define a \emph{$T$-chain} as any sequence 
$\{\vr 0\} = V_0 \subset V_1 \subset \ldots \subset V_\ell = \spn T$
of subspaces of $\spn T$ where each $V_i$ is spanned by some subset of $T$,
and $\dim{V_1} = 1$.
There are only finitely many $T$-chains, since $T$ is finite.
Given a $T$-chain $F = (V_0 \subseteq V_1 \subseteq \ldots \subseteq V_\ell)$
and a vector $\vr x \in \spn T$, we define the \emph{$F$-projection} of $\vr x$ 
as the sequence 
$(\vr x^F_1, \vr x^F_2, \ldots, \vr x^F_\ell)$ where 
\[
\proj{V_{i-1}}(\vr x) + \vr x^F_i \ = \ \proj{V_i}(\vr x)
\]
for every $i \in [\ell]$.
Thus $\vr x^F_i$ is orthogonal to $V_{i-1}$ and belongs to $V_i$.
We define the set 
\[
L_F \ := \ \setof{\vr x \in \spn T}
{\norm{\vr x^F_i} = r_{\dim {V_i}} \text{ for all } i\in [\ell]}.  
\]
In particular, as $\dim{V_1}=1$ and $\dim{V_\ell} = \dim{\spn T} = d$, 
every vector $\vr x \in L_F$ satisfies $\norm{\vr x^F_1} = r_1 = 1$ and
$\norm{\vr x^F_\ell} = r_d$.
The sets $L_F$ allow us to give the following equivalent definition of $\KK T$:

\begin{claim} \label{claim:equiv}
$\KK T = \convhull \Big(\bigcup_{F} L_F\Big)$, where the union ranges 
over all $T$-chains $F$.
\end{claim}
\begin{claimproof}
By induction on $d=\dim{\spn T}$.
If $d=1$, there is only one $T$-chain 
$F = (\{\vr 0\} \subset \spn {\vr t})$,
$L_F = \{-\vr t, \vr t\}$,
and $\KK T = \interval{-\vr t}{\vr t} = \convhull(L_F)$.

Suppose now that $d\geq 2$.
By induction assumption, for every $U\in {\cal U}_T$ we have
$\KK U = \convhull \Big(\bigcup_{F'} L_{F'}\Big)$, where the union ranges 
over all $U$-chains $F'$.
Thus, by \eqref{eq:L},
\[L_U \ = \ R_U + \convhull \Big(\bigcup_{F'} L_{F'}\Big) \ = \ 
\convhull \Big(\bigcup_{F'} \big(R_U + L_{F'}\big)\Big).\]
Any $U$-chain $F'$ can be extended to a $T$-chain $F$ by adding $\spn T$ at the end of $F'$,
and $L_{F'}$ and $L_F$ are related by $L_F = R_U + L_{F'}$.
Conversely, every $T$-chain $F$ extends an $U$-chain $F'$ in this way, 
for some $U\in{\cal U}_T$, and therefore
$\KK T = \convhull \Big(\bigcup_{U\in {\cal U}_T} L_U\Big) = \convhull \Big(\bigcup_{F} L_F\Big)$, 
where the union ranges over all $T$-chains $F$.
\end{claimproof}

\para{Basic geometric properties}


For every choice of $U\in {\cal U}_T$ and $\vr r \in R_U$, the set $\KK T$ includes not only 
the two shifted copies
$\vr r + \KK U$ and $-\vr r + \KK U$ but, being the convex hull, includes also all the `shifted copies'
in-between these two ones, namely all 
$\vr r' + \KK U$ where $\vr r' \in \interval {-\vr r}{\vr r}$:
\begin{lemma} \label{lem:int}
For every $U\in{\cal U}_T$ and $\vr r \in R_U$, we have
$\interval {-\vr r}{\vr r} + \KK U \subseteq \KK T$.
\end{lemma}

For every $U\in{\cal U}_T$ and $\vr r \in R_U$, the projection of $\vr r + \KK U$ on
$\spn U$ is exactly $\KK U$, since $\vr r$ is orthogonal to $\spn U$. 
Our next claim is a slight generalization of this observation,
relying on the fact that Minkowski sum commutes with projection:

\begin{lemma}\label{lem:proj}
For every $U, U'\in{\cal U}_T$ with $U\subseteq U'$ and $\vr r \in R_U$, 
$
\proj{\spn{U'}}(\vr r + \KK U) \ = \ \proj{\spn{U'}}(\vr r) + \KK U.
$
\end{lemma}
%



\Cref{lem:dist1} reinterprets the value of $\eps$ in terms of distances 
from a point to $\spn U$ where $U\subset T$, and to $\spn {\vr t}$, where
$\vr t \in T \setminus \spn U$.
It follows by Def.~\ref{def:delta} and Fig.~\ref{fig:triangle} (left).
Its corollary, \Cref{lem:dist}, obtained by
taking orthogonal complements of $\spn U$ and $\spn{\vr t}$,
is crucial for correctness of our construction.
\begin{lemma} \label{lem:dist1}
Let $U\subset T$, $\vr t\in T\setminus \spn U$, and $\vr x \in \spn T$.
If one of $\dist{\spn U, \vr x}$, 
$\dist{\spn {\vr t}, \vr x}$ is smaller than $\eps\norm{\vr x}$ then
the other one is larger than $\eps\norm{\vr x}$.
\end{lemma}
%
%
\begin{lemma} \label{lem:dist}
Let $U\subseteq T$,
$\vr t\in T\setminus \spn U$, and
$\vr x \in \spn {T}$, 
assuming $\dim{\spn U} +1 = \dim{\spn{T}}$. 
If one of $\dist{\hypp{U} \cap \spn T, \vr x}$, 
$\dist{\hypp {\vr t} \cap \spn{T}, \vr x}$ is smaller than $\eps\norm{\vr x}$ then
the other one is larger than $\eps\norm{\vr x}$.
\end{lemma}

\begin{proof}
Since $\dim{\spn U} +1 = \dim{\spn{T}}$, the subspace $\hypp{U}\cap \spn T$ is
a one-dimensional line.
According to \Cref{lem:dist1},
for every $\vr x \in \spn T$,
if one of $\dist{\spn U, \vr x}$, $\dist{\spn {\vr t}, \vr x}$ is 
smaller than $\eps\cdot\norm{\vr x}$ then
the other one is larger than $\eps\cdot\norm{\vr x}$.
We deduce \Cref{lem:dist}  using two orthogonality relations:
the hyperplane $\spn U$ is orthogonal to the line $\hypp{U}\cap \spn T$, 
and the hyperplane $\hypp {\vr t}\cap\spn T$ is orthogonal to the line 
$\spn{\vr t}$.
Therefore the angle between the line $\spn{\vr t}$ and 
the hyperplane $\spn U$ (the first pair) is exactly the same as the angle
between the line $\hypp U \cap \spn T$ and the 
hyperplane $\hypp {\vr t}\cap\spn T$ (the second pair).
Therefore there is an isometry of $\spn T$ 
that maps the first pair to the second one.
In consequence, the property of distances of $\vr x \in \spn T$
to the first pair, 
given by Lemma \ref{lem:dist1}, carries over to the second pair:
for every $\vr x \in \spn T$,
if one of $\dist{\hypp{U} \cap \spn T, \vr x}$, 
$\dist{\hypp {\vr t} \cap \spn{T}, \vr x}$ is smaller than $\eps\norm{\vr x}$ then
the other one is larger than $\eps\norm{\vr x}$.
This completes the proof.
\end{proof}

\Cref{lem:dist1,lem:dist} are stated for strict inequalities, 
but reading them in contrapositive yields the versions with
non-strict inequalities holding true as well.

\subsection{The sets $\KK T$ are $T$-projective} \label{sec:correctness}

We now show the correctness of our construction, that is, we prove 
that $\KK T$ is $T$-projective, by induction on $d =\dim{ \spn T}$.



\para{Induction base: $d=1$} 
The line segment $\KK {\vr t} = \interval {-\vr t}{\vr t}$ is readily seen to be
$\{-\vr t, \vr t\}$-projective.

\para{Warm-up: $d=2$} 

We argue that $\KK T$ is $T$-projective.
Consider an arbitrary $\vr t \in T$.
By Lemma \ref{lem:int} we get $ \interval {-\vr r_{\vr t}}{\vr r_{\vr t}} + \interval{-\vr t}{\vr t} \subseteq \KK T$, 
and therefore it is enough to prove 
\[
\proj{\vr t}(\KK T) \ \subseteq \ \interval {-\vr r_{\vr t}}{\vr r_{\vr t}}.
\]
As $\KK T$ is a convex hull of line segments $K'_{\vr u}$ and $K''_{\vr u}$, for $\vr u \in T$, it is enough
to prove that the projection $\proj{\vr t}$ maps all the line segments into 
$\interval {-\vr r_{\vr t}}{\vr r_{\vr t}}$.
This means that for every $\vr u \in T$ and $\vr x \in \KK {\vr u} = \interval{-\vr u}{\vr u}$ we have to prove 
\begin{claim} \label{claim:triangle}
$\proj{\vr t}(\vr r_{\vr u} + \vr x) \in \interval {-\vr r_{\vr t}}{\vr r_{\vr t}}$ \ and \ 
$\proj{\vr t}(-\vr r_{\vr u} + \vr x) \in \interval {-\vr r_{\vr t}}{\vr r_{\vr t}}$.
\end{claim}
\begin{claimproof}
We focus on the former claim, relying on the apparent symmetry.
Knowing that 
$\dist{\spn{\vr r_{\vr u}}, \vr r_{\vr u} + \vr x} = \norm{\vr x} = 1 = 
\eps \norm{\vr r_{\vr u} + \vr x}$, 
we apply the non-strict version of Lemma \ref{lem:dist} 
(applied to $d=2$, $U=\{\vr u\}$, $\hypp U \cap \spn T = \spn{\vr r_{\vr u}}$,
$\hypp{\vr t}\cap \spn T = \spn{\vr r_{\vr t}}$)
to get $\dist{\spn{\vr r_{\vr t}}, \vr r_{\vr u} + \vr x} \geq 1$,
which implies, by Fig.~\ref{fig:triangle} (right) in case $d=2$, that
$\proj{\vr t}(\vr r_{\vr u} + \vr x) \in \interval {-\vr r_{\vr t}}{\vr r_{\vr t}}$.
\end{claimproof}

\para{Inductive step: $d\geq 2$} 

Relying on the inductive assumption that the sets $\KK U$ are $U$-projective for all $U\in {\cal U}_T$,
we argue that $\KK T$ is $T$-projective.
We pick up an arbitrary $\vr t \in T$, and aim at proving
the inclusion $\proj{\vr t}(\KK T) + \interval {-\vr t}{\vr t} \subseteq \KK T$.
By convexity of $\KK T$ it suffices to prove, for every $U\in {\cal U}_T$ and
every vector $\vr r\in R_U$, the inclusion:
\begin{align} \label{eq:topro}
\proj{\vr t}(\vr r + \KK U) + \interval {-\vr t}{\vr t} \ \subseteq \ \KK T.
\end{align}
We distinguish two cases, depending on whether $\vr t$ belongs to $U$ or not.
\para{Case 1:  $\vr t \in U$}
Since $\vr t$ is orthogonal to $\vr r$, we have
\[
\proj{\vr t}(\vr r + \KK U) + \interval{-\vr t}{\vr t}  \ = \ 
\vr r + \proj{\vr t}(\KK U) + \interval{-\vr t}{\vr t},
\]
and since the set $\KK U$ is $U$-projective by inductive assumption, the right-hand side is included in
$
\vr r + \KK U \ \subseteq \ \KK T,
$
which implies \eqref{eq:topro}.

\para{Case 2:  $\vr t \notin U$}
%
%
%
Let $B_U := \slball{\vr 0}{R_e}\cap \spn U$ be the ball of radius $R_e$ inside $\spn U$,
where $e=\dim{\spn U}$.
Due to Lemma \ref{lem:B} we have $\KK U\subseteq  B_U$, and therefore 
$\vr r + \KK U\subseteq  \vr r + B_U$.
We distinguish two subcases, depending on whether the distance
$\dist{\hypp{\vr t}, \vr r + B_U}$ is smaller or larger than $R_e$.
Note that
$\dist{\hypp {\vr t}, \vr x} \ = \ 
\dist{\hypp {\vr t} \cap \spn T, \vr x}$ for every $\vr x \in \spn T$.

\para{Case 2a:  $\dist{\hypp {\vr t}, \vr r + B_U} \geq R_{e}$}

%
We prove the following inclusion:
\begin{align} \label{eq:f}
\proj{\vr t}(\vr r + \KK U) \ \subseteq \  Y \ := \ 
\hypp {\vr t} \cap \spn T \cap \slball {\vr 0}{r_d}.
\end{align}
Pick up an arbitrary point $\vr u \in \KK U\subseteq B_U$.
Let $\vr r' \ := \ \proj{\vr t}(\vr r + \vr u)$.
If $\vr r'=\vr 0$, we immediately deduce $\proj{\vr t}(\vr r + \vr u) \in Y$
since $\vr 0\in Y$.
Otherwise,
consider two lines: $\spn{\vr r}$ and $\spn{\vr r'}$.
Relying on $\vr u \in B_U$, $R_e \leq \dist{\hypp {\vr t}, \vr r + B_U}$
and $\vr r' \in \hypp{\vr t}$, respectively, 
we get the chain of three inequalities: 
\[
\dist{\spn{\vr r}, \vr r + \vr u} \ = \ \norm{\vr u} \ \leq \ R_e \ \leq \  
\dist{\hypp {\vr t}, \vr r + B_U} \ \leq \ \dist{\spn{\vr r'}, \vr r + \vr u}.
\]
Therefore the projection of 
$\vr r + \vr u$ on $\spn{\vr r'}$ has norm smaller or equal to the norm
of the projection of $\vr r + \vr u$ on $\spn{\vr r}$.
The former projection is $\vr r'$ while the latter one is $\vr r$, 
of norm $\norm{\vr r} = r_d$,
which implies $\norm{\vr r'} \leq r_d$ and hence 
$\proj{\vr t}(\vr r + \vr u) \in Y$.
This proves \eqref{eq:f}.

Using \eqref{eq:f}, we deduce \eqref{eq:topro}:
\[
\proj{\vr t}(\vr r + \KK U) + \interval{-\vr t}{\vr t} \ \subseteq \  
\big(\hypp {\vr t} \cap \spn T \cap \slball {\vr 0}{r_d}\big) + \interval{-\vr t}{\vr t}  
\ = \ \convhull(L_{\vr t}) \ \subseteq \ \KK T.
\]
\vspace{-4mm}

\para{Case 2b:  $\dist{\hypp {\vr t}, \vr r + B_U} < R_{e}$}

Let $U'$ be the inclusion-maximal subset of $T$ with $\spn{U'}  
= \spn{U \cup \{\vr t\}}$.
We start by observing that $e \leq d-2$. 
Towards contradiction, suppose $e = d-1$, and take any $\vr u \in B_U$ such that
 $\dist{\hypp {\vr t}, \vr r + \vr u} = 
 \dist{\hypp{\vr t}\cap\spn T, \vr r + \vr u}
 < R_{d-1} = \eps \norm{\vr r + \vr u}$.
By Lemma \ref{lem:dist} 
(since $\dim{\spn U}=e=d-1$, we have $\hypp U \cap \spn T = \spn{\vr r}$)
we deduce that
 $\dist{\spn{\vr r}, \vr r + \vr u} = \norm{\vr u} > R_{d-1}$, a contradiction.
Thus 
$\dim{\spn{U'}} = e+1\leq d-1$.
In consequence, $U' \in {\cal U}_T$.
%

%
\begin{claim} \label{lem:c2}
$\vr r + \KK U  \subseteq \convhull(L_{U'})$.
\end{claim}
We apply Case 1 to $U'$ (as $\vr t \in U'\in {\cal U}_T$), and obtain
$
\proj{\vr t}(\vr r' + \KK {U'})  + \interval{-\vr t}{\vr t} \ \subseteq \ \KK T,
$
for every $\vr r' \in R_{U'}$, which implies
$
\proj{\vr t}(\convhull(L_{U'}))  + \interval{-\vr t}{\vr t} \ \subseteq \ 
\KK T.
$ 
By \Cref{lem:c2} we get
$
\proj{\vr t}(\vr r + \KK U)  + \interval{-\vr t}{\vr t} \ \subseteq \ 
\proj{\vr t}(\convhull(L_{U'}))  + \interval{-\vr t}{\vr t}.
$ 
Composing the two last inclusions yields \eqref{eq:topro}, namely:
$\proj{\vr t}(\vr r + \KK U)  + \interval{-\vr t}{\vr t} \ \subseteq \ 
\proj{\vr t}(\convhull(L_{U'}))  + \interval{-\vr t}{\vr t}
\ \subseteq \ \KK T$.

\begin{proof}[Proof of \Cref{lem:c2}]
Consider the projection of $\spn T$ onto the subspace $\spn{U'}$, and
let $\vr r' = \proj{\spn{U'}}(\vr r)$.
By Lemma \ref{lem:proj} we get
\begin{align} \label{eq:projproj}
\proj{\spn{U'}}(\vr r + \KK U) \ = \ \vr r' + \KK U.
\end{align}
Furthemore, as $\vr r$ is orthogonal to $\spn U$, its projection $\vr r'$ is also so.

By assumption, $\dist{\hypp {\vr t}, \vr r + B_U} < R_{e}$, which means that
there is a point $\vr u \in B_U$ with $\dist{\hypp{\vr t},\vr r+ \vr u} < R_e$.
W.l.o.g.~we may assume that $\vr u$ belongs to the sphere of $B_U$.
Indeed, if $\vr u$ is in the interior of $B_U$, 
perturb it slightly, if necessary, to get $\vr u\neq \vr 0$ and not colinear with $\vr t$.
Then scale it positively
and negatively to get two colinear radiuses of $B_U$ forming a line segment of length
$2R_e$. As the distance of the line segment from $\hypp{\vr t}$ is smaller than $R_e$,
and the line segment is not orthogonal to $\hypp{\vr t}$,
one of its ends is necessarily at distance smaller than $R_e$ from $\hypp{\vr t}$.
Thus we may replace $\vr u$ by that endpoint. 
In particular, $\norm{\vr u} = R_e$.

Consider the distance of $\vr r' + \vr u\in\spn{U'}$ to 
$\spn{\vr r'} = \hypp{U}\cap \spn{U'}$ and to $\hypp{\vr t}\cap\spn{U'}$.
We claim:
\begin{align} \label{eq:dist2}
\dist{\hypp{U}\cap \spn{U'}, \vr r' + \vr u} \ = \ R_e
\qquad\qquad
\dist{\hypp {\vr t} \cap \spn{U'}, \vr r' + \vr u} \ < \ R_e.
\end{align}
The left equality follows since $\vr r'$ and $\vr u$ are orthogonal, and therefore
$\dist{\spn{\vr r'}, \vr r' + \vr u} \ = \ \norm{\vr u}$.
To prove the right inequality, we notice that $\vr t \in U'$ and therefore the projection on $\spn{U'}$ 
does not change the distance to $\hypp{\vr t}$:
\[
\dist{\hypp {\vr t}\cap\spn{U'}, \vr r' + \vr u} \ = \ 
\dist{\hypp {\vr t}, \vr r + \vr u}
\]
while,
by assumption, we have $\dist{\hypp {\vr t}, \vr r + \vr u} < R_e$.
Having \eqref{eq:dist2}, we apply Lemma \ref{lem:dist} to get
$R_e \geq \eps  \norm{\vr r' + \vr u}$, which rewrites to
$\norm{\vr r' + \vr u}\leq R_e/\eps  = R_{e+1}$.
As $\norm{\vr u} = R_e$, relying on
Pythagoras (cf.~Fig.~\ref{fig:triangle}, right) we obtain
$\norm{\vr r'}\leq r_{e+1}$, and therefore using Lemma \ref{lem:int} we get
\[
\vr r' + \KK U  \ \subseteq \ 
\KK {U'}.
\]
%
Decompose $\vr r = \vr r' + \vr r''$, where $\vr r''$ is orthogonal to $\spn{U'}$, and rewrite
add $\vr r''$ to both sides of
the above inclusion:
\[
\vr r + \KK U \ = \ \vr r'' + \big(\vr r' + \KK U \big) \ \subseteq \ \vr r'' + \KK {U'}.
\]
Since $\norm{\vr r''} \leq \norm{\vr r} = r_d$, and $\vr r''$ 
is orthogonal to $\KK {U'}$,
we have
$\vr r'' + \KK {U'} \subseteq \convhull(L_{U'})$,
and in consequence
$
\vr r + \KK U  \ \subseteq \ \convhull(L_{U'})
$, as required.
\end{proof}

\section{Generalization to Vector Partitioning}
\label{sec:partition}

%
%
%

In this section, we consider the online vector balancing game from the introduction. We prove the same bounds on Chooser's greedy algorithm for this game as in the previous variant. Similarly to the 2-partition balancing game, Chooser should minimize the value  
$\mathrm{val}(J_1,\ldots, J_p)$ with the assignment of the vector $\mathbf t_n$, given the assignment of $\mathbf t_1,\ldots, \mathbf t_{n-1}$. 
To this end, define $\mathbf a_i^{(n)}:=\sum_{j\in J_i\cap [n]}\mathbf t_i$
 as the partial sum of the vectors indexed by $j\in J_i, j\leq n$.
We further define the matrix $A^{(n)}=\left(\mathbf a_1^{(n)} \hdots\ \mathbf a_p^{(n)}\right)\in \R^{d\times p}$ with the $i$-th column $\mathbf a_i^{(j)}$ for all $i\in \{1,\ldots, p\}$. Note that $\mathrm{val}(J_1,\ldots,J_p)=\displaystyle\max_{j\in [n]}\displaystyle\max_{i_1\in [p]}\displaystyle\max_{i_2\in [p]}\|\mathbf a_{i_1}^{(n)}-\mathbf a_{i_2}^{(n)}\|$.
 \begin{definition}
    Let $n\in \N_{\geq 0}$. A sequence $A^{(0)}=0^{d\times p}, A^{(1)},\ldots, A^{(n)}$ of $(d\times p)$-matrices is a \emph{greedy $T$-matrix sequence} if either $n=0$, or 
    \begin{itemize}
        \item $A^{(0)}=0^{d\times p}, A^{(1)},\ldots, A^{(n-1)}$ is a greedy $T$-matrix sequence, and 
        \begin{itemize}
            \item either $A^{(n-1)}=A^{(n)}$,
               \item or there exists some $\mathbf t_j\in T$ such that 
$         A^{(n)}_i= A^{(n-1)}_i+  \mathbf t_n \cdot \mathbf e_i^\top$
    where $n\in J_{i'(j)}$ and 
     $ i'(j)\in
     \mathrm{argmin}_{i\in[p]}\langle \mathbf a_i^{(n-1)},\mathbf t_j\rangle$.
        \end{itemize}
     
    \end{itemize}

\end{definition}
The reader might notice that the first condition in the above definition does not have an analogue in the $2$-partition setting. Indeed, this additional condition will ensure us to observe that a restriction of a greedy $T$-matrix sequence is again a greedy $T$-matrix sequence, which would otherwise have not been true (cf.~\cref{lemma:restrictiontovecs}).

A greedy partition, as defined in the introduction, gives rise to a greedy $T$-matrix sequence. In accordance with this, the notion of greedy $T$-matrices extend the original problem setting. 

\begin{observation}\label{obs:matrices2columns}If 
 $A^{(0)},\ldots, A^{(n)}\in \R^{d\times 2}$ is a greedy $T$-matrix sequence, then 
the sequence $(\mathbf a_1^{(j)}-\mathbf a_2^{(j)})_{\{j: A^{(j)}\neq A^{(j-1)}\}}$ is a greedy $T$-sequence.

\end{observation}

\begin{lemma}\label{lemma:restrictiontovecs}
    Let $A^{(0)},\ldots, A^{(n)}$ be a greedy $T$-matrix sequence, and $i_1\neq i_2\in [p]$. Then 
the sequence 
    $(\mathbf a_{i_1}^{(0)}\mathbf a_{i_2}^{(0)}),\ldots, (\mathbf a_{i_1}^{(n)}\mathbf a_{i_2}^{(n)})$  of $(d\times 2)$-matrices is a greedy $T$-matrix sequence.
\end{lemma}
\begin{proof}
    We prove the statement by induction. For $n=0$, the claim is correct. We assume that $(\mathbf a_{i_1}^{(0)}\mathbf a_{i_2}^{(0)}),\ldots, (\mathbf a_{i_1}^{(n-1)}\mathbf a_{i_2}^{(n-1)})$ is a greedy $T$-matrix subsequence.
    
      Unless $A^{(n)}=A^{(n-1)}+\mathbf t_n\cdot \mathbf e_{i'(j)}$ with $i'(j)\in \{i_1,i_2\}$, we have $(\mathbf a_{i_1}^{(n-1)}\mathbf a_{i_2}^{(n-1)})=(\mathbf a_{i_1}^{(n)}\mathbf a_{i_2}^{(n)})$ and the statement is valid. Thus, we may assume $A^{(n)}=A^{(n-1)}+\mathbf t_n\cdot \mathbf e_{i'(j)}$ with $i'(j)=i_1$ without loss of generality. Because $A^{(i)}$ is a greedy $T$-matrix sequence, it must hold that 
    $\langle \mathbf a_{i_1},\mathbf t_j\rangle\leq \langle \mathbf a_{i_2},\mathbf t_j\rangle$ and the inductive step follows.
\end{proof}

By the above observation combined with \cref{lemma:restrictiontovecs}, we obtain that the discrepancy between any pair of partitions is that of a $T$-greedy sequence, which is bounded by \cref{thm:main}. Thus, we are ready to conclude the discrepancy bounds for greedy partitioning.

\greedypartition*
\begin{proof}
    We define $\mathbf a_i^j$ as well as $A^{(j)}$ as above for $i\in [p]$, $j\in \{0\}\cup[n]$. For all fixed ($j,i_1,i_2$) that attain $\mathrm{val}(J_1,\ldots,J_p)=\max_{j\in [n]}\max_{i_1\in [p]}\max_{i_2\in [p]}\|\mathbf a_{i_1}^{(n)}-\mathbf a_{i_2}^{(n)}\|$, we have 
    $\|\mathbf a_{i_1}^{(j)}-\mathbf a_{i_2}^{(j)}\|\leq G(T)$.
    Indeed, by Lemma \ref{lemma:restrictiontovecs}, the sequence $(\mathbf a_{i_1}^{(0)}\mathbf a_{i_2}^{(0)}),\ldots, (\mathbf a_{i_1}^{(n)}\mathbf a_{i_2}^{(n)})$ is a greedy $T$-matrix sequence, and by Observation \ref{obs:matrices2columns}, the sequence  $(\mathbf a_1^{(j)}-\mathbf a_2^{(j)})_{\{j: A^{(j)}\neq A^{(j-1)}\}}$ is a $T$-greedy sequence which includes every distinct value of $(\mathbf a_{i_1}^{(j)}-\mathbf a_{i_2}^{(j)})$. Therefore, the estimation follows by Theorem \ref{thm:main}.
\end{proof}

 
\section{Application to Scenario Scheduling}
\label{sec:scheduling}

In this section, we show that Conjecture \ref{conj:minsum} for the problem
\textsc{MinAvgSTC($d$)} from the introduction is true when the weights
are exponentially decreasing. Precisely, we give the proof of
\Cref{thm:scheduling-conj}, which yields the validity of the conjecture when
the weights satisfy $\frac{w_j}{w_{j+1}}\geq 2d$ for all $j\in [n-1]$.

Throughout the discussion, we will consider partial assignments of jobs to machines, or equivalently, partitions of $[j]$ for some $j\in [n]$. We denote such partitions by the tuple $(J_1\cap [j], \ldots, J_m\cap [j])$, that is, as the restriction of a full partition $(J_1,\ldots, J_m)$. We start by expressing the terms $|J_i\cap S_k\cap [j]|$ in terms of vectors.
\begin{definition}
    Let $j\in [n]$. For a fixed partial assignment $(J_1\cap [j],\ldots, J_m\cap [j])$, define $\mathbf a_i^{(j)}(J_1,\ldots,J_m)=(a^{ij}_k)_{\{k\in [d]\}}$ with $a^{ij}_k=|J_i\cap S_k\cap [j]|$.
\end{definition}
For the defined vectors, we observe that  optimizing the objective function \eqref{eq:schedobj} while assigning a single job corresponds to the greedy algorithm of the vector balancing game on $p=m$ partitions.
Here, the vector that arrives is the incidence vector of the job $j$ in the respective scenarios, and a machine with load vector $\mathbf a_i^{(j)}$ minimizing the scalar product is chosen.

\begin{lemma}\label{lem:schedisbalancing}
 Let $(J_1\cap [j-1],\ldots, J_m\cap [j-1])$ be a partial schedule. Any extension $(J_1\cap [j],\ldots, J_m\cap [j])$ of this schedule that minimizes $\mathrm{obj}(J_1\cap [j],\ldots, J_m\cap [j])$, satisfies $j\in J_{i'}$ for some
\[i'\in\mathrm{argmin}_{i\in[m]}\langle \mathbf a_i^{j-1}, \mathbf t^{j}\rangle
\text{ with }\ \mathbf t^{j}_k=\begin{cases}
    1, &j\in S_k\\
    0, &\text{otherwise}.
\end{cases}\]
\end{lemma}
\begin{proof}
   Let $i'$ be the machine that the job $j$ is assigned to. The objective values $ \mathrm{obj}(J_1\cap [j],\ldots,J_m\cap [j])$ and $ \mathrm{obj}(J_1\cap [j-1],\ldots,J_m\cap [j-1])$ differ by 
    \[
       \sum_{k: j\in  S_k}w_{j}\cdot|J_{i'}\cap S_k\cap[j]]|=w_{j}\cdot\langle \mathbf a_{i'}^{(j)}, \mathbf t^{j}\rangle=w_{j}\cdot\langle \mathbf a_{i'}^{(j-1)}+\mathbf t^{j}, \mathbf t^{j}\rangle=w_{j}(\langle \mathbf a_{i'}^{(j-1)}, \mathbf t^{j}\rangle+\|\mathbf t^{j}\|^2).
       \]
       Here, the terms $w_{j}$ and $\|\mathbf t^{j}\|^2$ do not depend on the choice of the assignment, so the objective function is minimized if and only if $\langle \mathbf a_{i'}^{(j-1)}, \mathbf t^{j}\rangle$ is minimized.
\end{proof}

Next, we argue that if $w_j\geq 2dw_{j+1}$ for all $j\in [n-1]$, then every optimal solution must follow a greedy trajectory, i.e., in any optimal solution that minimizes $\mathrm{obj}(J_1,\ldots,J_m)$, the $j$-th job is placed so that $\mathrm{obj}(J_1\cap [j],\ldots, J_m\cap [j])$ is as small as possible given the fixed partitioning of the first $j-1$ jobs.

\begin{lemma}\label{lem:optislex}
    Let $n$ jobs be given such that $w_j\geq 2dw_{j+1}$ for all $j\in [n-1]$ and let $(J_1,\ldots, J_m)$ be an optimal partition of the jobs with respect to $\mathrm{obj}$. Then for every $j\in [n]$, we have
    $j\in J_i$  with  $i\in\mathrm{argmin}_{i\in[m]}\{\mathrm{obj}(J_1\cap [j],\ldots J_m\cap [j]\},$
    given $(J_1\cap [j-1],\ldots J_m\cap [j-1])$.
\end{lemma}
\begin{proof}
    Assume that this is not the case. Fix a partition $(J_1,\ldots, J_m)$ that minimizes $\mathrm{obj}$ and let $j$ be the smallest index such that $(J_1\cap[j],\ldots, J_m\cap [j])$ is not optimal  with respect to $\mathrm{obj}$. Then, by the minimality of $j$, $(J_1\cap[j-1],\ldots, J_m\cap [j-1])$ is optimal and there is a machine $i'\in [m]$ such that 
    \begin{align}\label{eq:primebetter}
    \mathrm{obj}\left(J_1\cap [j-1],\ldots, (J_{i'}\cap[j-1])\cup\{j\},\ldots, J_m\cap [j-1]\right)<\mathrm{obj}(J_1\cap [j],\ldots, J_m\cap [j]).
\end{align}
    We fix the assignment corresponding to the left hand side, that is, for the index $i'$ as above, we define $(J_1',\ldots, J_m')=(J_1\setminus\{j\},\ldots, J_{i'}\cup\{j\},\ldots, J_m\setminus\{j\})$. We would like to show that $\mathrm{obj}(J_1',\ldots, J_m')-\mathrm{obj}(J_1,\ldots, J_m)<0$, which contradicts the optimality of $(J_1,\ldots, J_m)$ and proves the statement.
      To this end, we consider the contribution of each job to the objective value. The first $j-1$ jobs contribute equally to the objective value in both schedules as they are assigned identically. Recall that $j\in J'_{i'}$ and assume further $j\in J_i$. For the contribution of the $j$-th job, we have 
\begin{align}
   \label{eq:start}\sum_{k\in [d]: j\in S_k} \hspace{-0.25cm}w_j\cdot|J'_{i}\cap S_k\cap [j]|&= w_j\cdot\sum_{k\in [d]: j\in S_k} \hspace{-0.25cm} |J'_i\cap S_k\cap[j]| \\&\leq w_j\cdot\left(\sum_{k\in [d]: j\in S_k} \hspace{-0.25cm} |J_i\cap S_k\cap [j]|-1\right)\\&=\sum_{k\in [d]: j\in S_k} w_j\cdot|J_i\cap S_k\cap[j]|-w_j.
\end{align}
    Here, the inequality follows by \eqref{eq:primebetter}, and the fact that both sums are integral. Finally, we analyze the contribution of jobs $j'\in \{j+1,\ldots, n\}$. We note that $j'\in J_i$ if and only if $j'\in J'_i$, and assuming this, compute 
   \begin{align}
   &\textcolor{white}{\leq } \sum_{k\in [d]: j'\in S_k} w_{j'}|J'_i\cap S_k\cap [j']|= w_{j'}\sum_{k\in [d]: j'\in S_k} |J'_i\cap S_k\cap [j']|\\&\geq w_{j'}\left(\sum_{k\in [d]: j'\in S_k} (|J_i\cap S_k\cap[j']|-1)\right)\geq\sum_{k\in [d]: j'\in S_k} w_{j'}|J_i\cap S_k\cap [j']|-dw_{j'}.\label{eq:end}
\end{align}
Here, the first inequality follows because $(J'_i\cap [j'])\setminus(J_i\cap [j'])\subseteq \{j\}$.
    Combining inequalities \eqref{eq:start}-\eqref{eq:end} yields 
    $\mathrm{obj}(J'_1,\ldots J'_m)-\mathrm{obj}(J_1,\ldots J_m)=-w_j+d(w_{j+1}+\ldots+ w_n)<0$,
    as claimed. The reader can verify the last inequality by induction, using $w_j\geq 2dw_{j+1}$ for all indices $j\in [n-1]$.
\end{proof}

We are now ready to prove the conjecture in the case of exponentially
decreasing weights.

\begin{proof}[Proof of \Cref{thm:scheduling-conj}]
    By Lemma \ref{lem:optislex}, every optimal partition $(J_1,\ldots, J_m)$ minimizes the value   $\mathrm{obj}(J_1\cap [j],\ldots J_m\cap [j])$ for all $j$, given $(J_1\cap [j-1],\ldots J_m\cap [j-1])$. Lemma \ref{lem:schedisbalancing} implies that $(\mathbf a_1^{(j)}\hdots \mathbf a_m^{(j)})_{j=0,\ldots,n}$ is a greedy $T$-sequence (with $\mathbf a_i^{(0)}=\mathbf 0$) and thus by Theorem \ref{thm:greedy-partition}, we have 
     \[\max_{j\in [n]}\|\displaystyle\max_{i\in [m]}\mathbf  a_{i_1}^{(j)}-\min_{i\in [m]}\mathbf a_{i_2}^{(j)}\|\leq G(\{0,1\}^d)\] 
     and thus
          \[\max_{j\in[n], k\in [K]}|\mathbf \displaystyle\max_{i\in [m]}\mathbf  a_{i_1}^{(j)}-\min_{i\in [m]}\mathbf a_{i_2}^{(j)}|\leq G(\{0,1\}^d)\] 
     where the right-hand side $g(d)\coloneqq G(\{0,1\}^d)$ only depends on $d$ and not on $n$, as desired.
\end{proof}

\section{Lower bound}
\label{sec:lower-bound}

In this section, we show a lower bound on $G(T)$ for $T \subseteq {\cal S}^{2d-1}$, 
namely we prove \Cref{lem:greedy-lower-bound}. 
The lower bound depends on the dimension $d=\dim{\spn T}$ and $\delta_T$.
\lowerbound*
The rest of this section is devoted to the proof of \Cref{lem:greedy-lower-bound}. 
We fix $d \in \N_+$ and $\delta \in \R_{\geq 0}$ and construct set 
$T \subseteq \R^{2d}$ with $\delta_T = \delta$ and $G(T) \geq \sqrt{d}/\delta$. 
Let $\vr e_i \in \R^{2d}$ be the vector with a one in the $i$-th coordinate and zero elsewhere. 
For $i \in [d]$ we define vectors $\vr u_i, \vr v_i$ as follows:
\begin{equation*}
    \vr u_i = \vr e_{2i} \quad \vr v_i = - \sqrt{1-\delta^2} \cdot \vr e_{2i} + \delta \cdot \vr e_{2i-1} 
\end{equation*}
The intuition behind this construction is as follows. By using $d$ pairs of vectors $\vr u_i$ and $\vr v_i$, we can ensure that $G(T) \geq \sqrt{d}/\delta$  Each such pair of vectors allows us to increment  the square of the norm of point reachable by a $[-1,1]T$-greedy sequence by $\frac{1}{\delta^2}$, which accumulates over all $d$ pairs. We define $T$ as follows:
\begin{equation*}
    T = \bigcup_{i \in [d]} \{\vr u_i, \vr v_i\}
\end{equation*}
The above idea follows from observing the behaviour of a two-dimensional greedy sequence
$x = \vr x_0, \vr x_1, \ldots, \vr x_n$ of vectors in $\R^2$ such that
\[
\vr x_{i+1} - \vr x_i \in \left\{
\sqrt{1-\delta^2}\,\vr u_1,\;
\vr v_1,\;
-\sqrt{1-\delta^2}\,\vr u_1,\;
-\vr v_1
\right\},
\]
and extending this construction to $d$ parallel copies. See \cref{fig:lower_bound} in the introduction for an illustration of the two dimensional case. 

%
%
%
%
%
%
%

\begin{claim}
    The set $T \subseteq {\cal S}^{2d-1}$ and $\delta_T = \delta$.
\end{claim}
\begin{proof}
Observe, that for every $i \in [d]$ we have $\|u_i\| = 1$ and $\|v_i\| = 1 - \delta^2 + \delta^2 = 1$. Hence $T \subseteq {\cal S}^{2d-1}$.

To show that $\delta_T = \delta$ we aim at proving:
\[
\min \{\dist{\spn U,\vr t}: U \subset T, \vr t \in T \setminus \spn U\} = \delta.
\]
Notice that the minimum on the left-hand side is always attained when $\dim{\spn U} = \dim{\spn T} - 1$. Since all $\vr v_1, \vr u_1, \vr v_2, \vr u_2 \ldots, \vr u_d, \vr v_d$ 
are linearly independent, we observe the following for every $t \in T$:
\begin{equation*}
\min \{\dist{\spn U,\vr t}: U \subset T, \vr t \in T \setminus \spn U\} = 
    \min \{\dist{\spn U,\vr t}: U = T \setminus \{\vr t, -\vr t\}\}
\end{equation*}
Observe, that for $i \in [d]$ and $\vr t \in \{\vr v_i, -\vr v_i\}$ we have:
\begin{equation*}
    \min \{\dist{\spn U,\vr t}: U = T \setminus \{\vr t, -\vr t\}\} =  \dist{\spn {\vr u_i},\vr t} = \delta 
\end{equation*}
Similarly, for $i \in [d]$ and $\vr t \in \{\vr u_i, -\vr u_i\}$ using the distance of a point to a line, we have:
\begin{align*}
    \min \{\dist{\spn U,\vr t}: U = T \setminus \{\vr t, -\vr t\}\} &=  \dist{\spn {\vr v_i},\vr t} = \\ &\sqrt{\|\vr t\|^2 - \frac{\innerprod{\vr t}{\vr v_i}^2}{\|\vr v_i\|}} = \sqrt{1 - (1-\delta^2)} = \delta 
\end{align*}

\end{proof}
Since $T \subseteq {\cal S}^{2d-1}$ and $\delta_T = \delta$ to finish the proof of \Cref{lem:greedy-lower-bound} it is enough to prove the following claim:
\begin{claim}\label{clm:absorbing_x}
    There exists a $[-1,1]T$-greedy sequence $\vr x_0, \vr x_1, \ldots, \vr x_n$ such that $\|\vr x_n\| \geq \frac{\sqrt{d}}{\delta}$.
\end{claim}
\begin{proof}
Let us define 
\begin{equation*}
    \vr w = \Sigma_{i \in [d]} \vr e_{2i-1}.
\end{equation*}
We observe the following claim, intuitively stating that along the line determined by the vector $\vr w$ we can find a greedy $[-1,1]T$-greedy sequence $\vr x_0, \vr x_1, \ldots, \vr x_n$ such $\vr x_n$ such that $\| \vr x_n \| \geq \frac{\sqrt{d}}{\delta}$.
\begin{claim}\label{clm:absorbing_multiplicity}
    For $l \in [0, \frac{1}{\delta} - \delta]$ if there exists a greedy $[-1,1]T$-sequence $\vr x_0, \vr x_1, \ldots, \vr x_n$ such that $l \vr w = \vr x_n$ then there exists a greedy $[-1,1]T$-sequence $\vr y_0, \vr y_1, \ldots, \vr y_m$ such that $(l+\epsilon) \vr w = \vr y_m$ for every $\epsilon \in (0, \delta]$. 
\end{claim}

Before we prove \Cref{clm:absorbing_multiplicity}, we show how it implies \Cref{clm:absorbing_x}. Since trivially there exists a single element $[-1,1]T$-greedy sequence ending in zero by repetitive use of \Cref{clm:absorbing_multiplicity} we get that there exists a $[-1,1]T$-greedy sequence ending in $\frac{1}{\delta} \vr w$, which has norm $\frac{\sqrt{d}}{\delta}$. This completes the proof of \Cref{clm:absorbing_x}.
\end{proof}

\begin{proof}[Proof of \Cref{clm:absorbing_multiplicity}]
Let us fix $l \in [0, \frac{1}{\delta} - \delta]$, $\epsilon \in (0, \delta]$ and let $\vr{x}_0, \vr{x}_1, \ldots, \vr{x}_n$ be a $[-1,1]T$-greedy sequence such that $l \vr w = \vr{x}_n$. We define a sequence of vectors $\vr y_0 = l \vr w$, $\vr y_{2i-1} = \vr y_{2i-2} + \frac{\epsilon}{\delta}\sqrt{1-\delta^2}\vr u_i$, and $\vr y_{2i} = \vr y_{2i-1} +  \frac{\epsilon}{\delta}\vr v_i$ for $i \in [d]$. We claim that $\vr{x}_0, \vr{x}_1, \ldots, \vr{x}_n, \vr{y}_1, \vr{y}_2, \ldots, \vr{y}_{2d}$ is a $[-1,1]T$-greedy sequence ending in $(l+\epsilon)\vr w$. Let us observe that for $i \in [d]$ we have $\vr y_{2i} = \vr y_{2i-2} +  \frac{\epsilon}{\delta}\sqrt{1-\delta^2}\vr u_i + \frac{\epsilon}{\delta}\vr v_i = \vr y_{2i-2} + \epsilon \cdot \vr e_{2i-1} = \vr y_0 + \sum_{i=1}^i \epsilon \cdot \vr e_i$. As $\vr y_0 = l \vr w$ we get that $\vr y_{2d} = (l+\epsilon) \vr w$. Hence, we only need to show that the sequence is $[-1,1]T$-greedy.

Observe that for $i \in [d]$ we have that $\innerprod{\vr y_{2i-2}}{\vr u_i} = 0$ and $\frac{\epsilon}{\delta}\sqrt{1-\delta^2}\vr u_i \in [-1,1]T$. Moreover, for $i \in [d]$ we have
\begin{equation*}
    \innerprod{\vr y_{2i-1}}{\vr v_i} = \innerprod{\frac{\epsilon}{\delta}\sqrt{1-\delta^2}\vr u_i + l \cdot \vr e_{2i-1}}{\frac{\epsilon}{\delta} \vr v_i} = -\frac{\epsilon}{\delta}(1-\delta^2) + l\epsilon \leq \frac{\epsilon}{\delta}(\delta^2 -1) + \frac{\epsilon}{\delta}(1 -\delta^2) = 0.
\end{equation*}
Moreover, $\frac{\epsilon}{\delta} \vr v_i \in [-1,1]T$. Hence, the sequence is $[-1,1]T$-greedy.  This completes the proof of \Cref{clm:absorbing_multiplicity}.

\end{proof}

\section{Concluding Remarks}

In this paper, we have shown that iterates of the Euclidean greedy vector
balancing algorithm remain universally bounded against sequences taken from any
finite set $T \subset {\cal S}^{d-1}$. Our upper bound $(2/\delta_T)^{d-1}$ and
lower bound $\sqrt{d/2}/\delta_T$ are unfortunately very far apart. A natural
question is thus to close this gap. We have also provided an application of our
greedy bound to scenario scheduling and SGD sample reordering. We however expect
there be many more interesting settings where greedy vector balancing bounds
can be applied, and we hope this will lead to a fertile area of future
research.   

\bibliographystyle{plain}
\bibliography{greedy}

\end{document}